\newif\iffull 
  \fulltrue

\documentclass[conference]{IEEEtran}

\def\BibTeX{{\rm B\kern-.05em{\sc i\kern-.025em b}\kern-.08em
		T\kern-.1667em\lower.7ex\hbox{E}\kern-.125emX}}

\usepackage{pgfkeys}
\usepackage{tikz}
\usetikzlibrary{shapes,arrows.meta,calendar,matrix,calc,fit,decorations,snakes}

\usepackage{amssymb}
\usepackage{cite}
\usepackage{algorithm}
\usepackage[noend]{algpseudocode}
\usepackage{todonotes}
\usepackage{amsthm}
\usepackage{amsmath}
\usepackage{amsfonts}
\usepackage{blindtext}
\usepackage{upquote}
\usepackage{float}
\usepackage{multirow}
\usepackage{booktabs}
\usepackage{colortbl}
\usepackage{listings}
\usepackage{mathtools}
\usepackage{mathpartir}
\newcommand{\nextrule}{\hspace{2em}}
\usepackage{subcaption}
\usepackage{zlmtt}
\usepackage{listings,caption}
\usepackage{wrapfig}
\usepackage{thmtools} 
\usepackage{thm-restate}
\usepackage{threeparttable}
\usepackage{stmaryrd}
\usepackage{xspace}
\usepackage{footmisc}
\usepackage{setspace}
\usepackage{float}
\usepackage{hhline}
\usepackage[colorlinks=true,urlcolor=black]{hyperref}

\usepackage{longtable}

\makeatletter
\renewcommand{\ALG@name}{Program}
\makeatother

\newtheorem{definition}{Definition}

\newtheorem{globaltheorem}{Theorem}
\newtheorem{globallemma}{Lemma}
\theoremstyle{definition}
\newtheorem{example}{Example}
\theoremstyle{plain}

\definecolor{blueC}{rgb}{0.0, 0.44, 1.0}
\definecolor{redC}{rgb}{0.8, 0.2, 0.2}

\newcommand{\benchmarkNumber}[0]{20\xspace}

\newcommand{\eval}[2]{{\llbracket{#1}\rrbracket^{#2}}}
\newcommand{\evalTable}[2]{{\llbracket{#1}\rrbracket^{#2}}}

\newcommand{\tp}[0]{tp}
\newcommand{\vl}[0]{vl}

\newcommand{\vw}[0]{v}
\newcommand{\tb}[0]{t}

\newcommand{\fv}[0]{fv}

\newcommand{\ie}[0]{i.e., }

\newcommand{\tuple}[1]{\langle #1 \rangle}

\newcommand{\powerset}[1]{\mathcal{P}{(#1)}}
\newcommand{\sep}[0]{ \ \mid \ }

\newcommand{\initM}[0]{m_{0}}

\newcommand{\users}[0]{\mathcal{U}}
\newcommand{\queries}[0]{\mathcal{Q}}

\newcommand{\mix}[0]{\mathrm{mix}}
\newcommand{\prg}[0]{\mathrm{prg}}
\newcommand{\stv}[0]{V}

\newcommand{\con}[0]{\mathrm{con}}
\newcommand{\dis}[0]{\mathrm{dis}}

\newcommand{\clEmpty}[0]{{\downarrow}}
\newcommand{\cl}[1]{{\downarrow}#1}
\newcommand{\clSet}[1]{{\downarrow}\{#1\}}
\newcommand{\tc}[0]{\mathrm{tc}}

\newcommand{\eq}[1]{{#1}_{\sim}}
\newcommand{\EQ}[1]{\llbracket{#1}\rrbracket}

\newcommand\Subseteq{\mathrel{\ooalign{$\subseteq$\cr\hidewidth\raise.225ex\hbox{$\Subset\mkern.5mu$}\cr}}}

\newcommand{\R}[0]{\mathbb{R}}
\renewcommand{\P}[0]{\mathbb{P}}
\renewcommand{\S}[0]{\mathbb{S}}
\newcommand{\Q}[0]{\mathbb{Q}}

\newcommand{\T}[0]{\mathbb{T}}

\newcommand{\proj}[2]{#1\negthickspace\downharpoonright_{#2}}
\newcommand{\mulStepEval}[2]{\xRightarrow{#1}\negthickspace_{#2}}

\newcommand{\toolname}[0]{\mbox{\textsc{DiVerT}}\xspace}

\newcommand{\tightpar}[1]{{\smallskip \noindent\bf #1. }}

\newcommand{\squote}[1]{\text{`}{\mathrm{#1}}\text{'}}

\newcommand{\stsubset}[0]{\sqsubseteq_{\mathrm{st}}}
\newcommand{\st}[0]{\mathrm{st}}

\newcommand{\itr}[0]{\mathrm{itr}}
\newcommand{\col}[0]{\mathrm{col}}
\newcommand{\dep}[0]{\mathrm{dep}}

\newcommand{\basicCodeStyle}{\ttfamily\small}
\newcommand{\keywordCodeStyle}{\bfseries}

\lstdefinelanguage{IMP}{
	keywords={if, while, then, else, input, out, do, skip, watch, write},
	sensitive=true,
	commentstyle=\small\itshape\ttfamily\textcolor{gray},
	keywordstyle=\textcolor{blue},
	identifierstyle=\ttfamily,
	basewidth={0.5em,0.5em},
	columns=fixed,
	comment=[l]{//},
	fontadjust=true,
	literate={=>}{{$\implies$}}3 {===}{{$\equiv$}}1 {=/=}{{$\not\equiv$}}1 {|>}{{$\triangleright$}}3 
	{\\/}{{$\vee$}}2 {/\\}{{$\wedge$}}2 {^}{{$\uparrow$}}1,
	morecomment=[s]{/*}{*/}
}

\makeatletter
\DeclareCaptionFormat{alsoempty}{%
	#1\if\relax\expandafter\noexpand\lst@caption\else#2#3\fi}
\makeatother

\begin{document}

\title{Disjunctive Policies for Database-Backed Programs}

\author{
	\IEEEauthorblockN{Amir M. Ahmadian}
	\IEEEauthorblockA{%
		\textit{KTH Royal Institute of Technology}\\
	}
	\and
	\IEEEauthorblockN{Matvey Soloviev}
	\IEEEauthorblockA{%
		\textit{KTH Royal Institute of Technology}\\
	}
	\and
	\IEEEauthorblockN{Musard Balliu}
	\IEEEauthorblockA{%
		\textit{KTH Royal Institute of Technology}\\
	}
}

\maketitle

\begin{abstract}
When specifying security policies for databases, 
it is often natural to formulate \emph{disjunctive} dependencies,
where a piece of information may depend on at most one of two dependencies
$P_1$ or $P_2$, but not both.
A formal semantic model of such disjunctive dependencies, the Quantale of Information,
was recently introduced by
Hunt and Sands
as a generalization of the Lattice of Information.
In this paper, we seek to contribute to the understanding of disjunctive dependencies
in database-backed programs and introduce a practical framework to statically
enforce disjunctive security policies. To that end, we 
introduce the \emph{Determinacy Quantale}, a new query-based structure 
which captures the ordering of disjunctive information in databases.
This structure can be understood as a query-based counterpart to the Quantale of Information.
Based on this structure, we design a sound enforcement mechanism to check
disjunctive policies for database-backed programs.
This mechanism is based on a type-based analysis
for a simple imperative language with database queries, which is
precise enough to accommodate a variety of row- and column-level database policies flexibly
while keeping track of disjunctions due to control flow.
We validate our mechanism by implementing it in a tool, \toolname, and demonstrate
its feasibility on a number of use cases.
\end{abstract}

\section{Introduction}\label{sec:introduction}

Database security and information flow security have largely evolved as two disparate areas~\cite{DBLP:journals/tdsc/BertinoS05,DBLP:dblp_journals/jsac/SabelfeldM03}, while sharing closely-related foundations and mechanisms to enforce security. Modern applications commonly rely on shared database backends to provide rich functionality to a multitude of mutually distrusting users. In response to frontend demands, database query languages, with features such as triggers, store procedures, and user-defined functions, have increasingly come to resemble
full-fledged programming languages, thus calling into question the adequacy of the underlying access control models~\cite{guarnieri2016strong,bender2014explainable}. 
A \emph{security policy} describes the totality of expectations that we have of a computer system in the face of adversaries that seek to satisfy objectives that may differ from ours. 
In the context of database systems, whose purpose is to retain and provide information, the security policies of interest constrain who is allowed to learn what parts of that information. 
A class of such security policies which has proven particularly challenging to enforce with the methods of database security are \emph{disjunctive policies}, which states that given two pieces of information, some entity may either learn one \emph{or} the other, but not both.

A common example of disjunctive policies are databases which contain personally identifiable information, such as medical trial data.
Biometric parameters of participants are important confounders that must be considered when drawing conclusions from the data, but at the same time releasing too many parameters of any one participant (such as their height, age and weight) might be sufficient to deanonymize them with high confidence~\cite{sweeney2002kanonymity}.
Hence, a security policy for such a database may specify that the user may learn height and age, or height and weight, or age and weight, but not all three.
Other examples of scenarios where disjunctive policies are useful include differential privacy~\cite{dwork2006differential} and secret sharing.

In this paper, we combine insights from database security and information flow research to develop a formal model for reasoning about disjunctive information in database-backed programs, and thus take a step towards reconciling the two fields.
Our model makes it possible to reason about the semantic information dependencies in a program that performs queries, and compare them against a disjunctive policy.
Building upon this, we propose a provably sound static enforcement mechanism that ensure that the policy is satisfied.

It is customary in information flow models to represent information as an equivalence relation on states, with the refinement order of equivalence relations corresponding to having more information.
This representation can be used for both the actual information conveyed by a computational process and the bound imposed on it as part of a simple, non-disjunctive security policy.
The possible equivalence relations on a given universe of states form a structure called the \emph{Lattice of Information} (LoI)~\cite{landauer1993lattice}, in which security-relevant questions can be answered, such as whether a program reveals no more information than is allowed by the security policy, or what information is revealed by the combination of two programs. Similar questions have been addressed in the database community using an analogous object called the Disclosure Lattice~\cite{bender2013fine}.
We observe that this definition is actually insufficient to characterize information, which motivates us to introduce a more specific structure based on query determinacy, the \emph{Determinacy Lattice} (DL).
The formal relation between the Disclosure Lattice or our definition and LoI was hitherto unexplored, and more importantly neither of them can be used to represent disjunctions as seen in our motivating example.

Recently, Hunt and Sands~\cite{hunt2021quantale} proposed a new information flow structure called the \emph{Quantale of Information} (QoI), which seeks to address this shortcoming and establish a formal setting for representing, combining and comparing disjunctions of information.
We build upon this work to introduce an analogous structure, the \emph{Determinacy Quantale} (DQ), representing disjunctive dependencies in database-backed programs.
As we show, this structure can be formally related to the QoI, and this relationship is analogous to that between the LoI and the DL. We then use the DQ to design a knowledge-based security condition that relates disjunctive dependencies in database-backed programs to disjunctive policies.

We are the first to address the problem of enforcing disjunctive policies.
Prior works that develop language-based enforcement techniques in database-backed applications do not support disjunctive policies, while database-level dependencies are restricted to coarse approximations that incorrectly reject secure programs, such as our previous example~\cite{chong2007sif,corcoran2009cross,balliu2016jslinq,DBLP:journals/pacmpl/ParkerVH19,guarnieri2019information}.

Perhaps unsurprisingly, path sensitivity of a static analysis is key to capturing disjunctive dependencies. We show how standard flow-sensitive type-based dependency analysis~\cite{delft2015very} can be adapted to a compositional path-sensitive analysis and thus capture disjunctive dependencies in terms of database queries. 
To represent these dependencies in the DQ model, we introduce a sound approximation of the information disclosed by each database query which is precise enough to represent complex combinations of both row- and column-level dependencies.
Finally, in the DQ, the combination of these analyses can be proven sound with respect to our security condition. 
We expect that the overall architecture of the resulting soundness proof, in which we relate a sequence of abstractions of the behaviour of a program to ordered elements of the DQ, can be generalized to many other enforcement mechanisms for our security condition.

To demonstrate the practicality of our approach, we implement this type-based dependency analysis and query approximation for database-backed programs and evaluate it on a test suite and some use cases which effectively illustrate the need for disjunctive dependencies and disjunctive policies.

{\iffull
	{}
\fi}

\tightpar{Summary of contributions}
\begin{itemize}
\item We introduce a formal model for reasoning about \mbox{disjunctive} dependencies and policies in databases.
In the process, we show how to reconcile perspectives from the database security and information flow communities.

\item We introduce a database-specific model of knowledge, the Determinacy Lattice, and a disjunctive extension, called the Determinacy Quantale, and explore their relationship to established general-purpose semantic models.
\item Using our model, we define an extensional security condition for database-backed programs that
accommodates disjunctive policies.
\item We propose a type-based program analysis to capture disjunctive dependencies in database-backed programs, combine them with a novel abstraction of queries, and prove them sound with respect to our security condition.
This is presented as an instance of a generalizable architecture for such soundness proofs.
	\item We implement a prototype tool that uses type-based dependency analysis and query approximation to verify query-based disjunctive policies for database-backed programs, and demonstrate its feasibility on a test suite and a number of use cases.
\end{itemize}
\vspace{4mm}

The rest of paper is structured as follows. After reviewing preliminaries in Section~\ref{sec:preliminaries}, we give our account of the DL and introduce the DQ in Section~\ref{sec:determinacy_quantale}. In Section~\ref{sec:security_model}, we formalize our model of database-backed programs and the security policies we impose on them, culminating in a formal security condition.
We present enforcement mechanisms in Section~\ref{sec:enforcement}, and their implementation and evaluation in Section~\ref{sec:implementation}. In Section~\ref{sec:related_work}, we contextualize our contributions with a discussion of related work, and finally summarize conclusions in Section~\ref{sec:conclusion}.

\section{Background}\label{sec:preliminaries}

\subsection{Lattice of Information}\label{sec:LoI}
An equivalence relation ${\sim} \subseteq A\times A$  on a set $A$ is a binary relation that is reflexive, symmetric, and transitive. For example, the equivalence relation \texttt{parity} on the set $A = \{ 0, 1 , 2 , 3 \}$ is defined as $\{ (x,y) \mid x,y \in A \wedge x \ mod \ 2 = y \ mod \ 2\}$.
An equivalence relation partitions its underlying domain into disjoint equivalence classes. Given an equivalence relation $P$ on a set $A$ and $a \in A$, $[a]_P$ denotes the unique equivalence class induced by $P$ that $a$ belongs to. We write $[P]$ to denote the set of all equi\-valence classes induced by $P$.
We call $[P]$ a \emph{partition} of $A$ and hereafter we may also refer to each element, i.e. equivalence class, of the partition $[P]$ as a \emph{cell}. For example, \texttt{parity} partitions $A$ into cells $\{ 0, 2\}$ and $\{1,3\}$.

Equivalence relations over states are commonly used to represent an agent's knowledge, by relating two states whenever the agent cannot distinguish between them.
When an equivalence relation models knowledge, we also call the cells induced by it \emph{knowledge sets}. These have a distinct intuitive interpretation when we consider functions $f$ that take in some state and return an agent's \emph{view} of it.
We will write the equivalence relation induced by the output of $f$ as $\sim_f = \{ (x,y) \mid f(x)=f(y) \}$.
In that case, in a state $a$, the knowledge set $[a]_{\sim_f}$ represents the agent's remaining uncertainty about the state, in the sense of all the states that the agent still considers possible, after observing the output of $f$. The agent \emph{knows} anything that is true in all states in the knowledge set.
In this paper, we use the terms knowledge and information interchangeably.

A complete lattice is a set equipped with a partial ordering (reflexive, antisymmetric, and transitive) relation, maximal and minimal elements $\top$ and $\bot$ for this relation and a join (least upper bound) for any subset of elements. The meet (greatest lower bound) of a subset can be defined as the join of the set of all lower bounds of that subset~\cite{kaplansky2001set}.
The \emph{Lattice of Information (LoI)} \cite{landauer1993lattice} is a structure for representing the ordering of information with equivalence relations. Let $\mathcal{L}(A)$ be the set of all equivalence relations defined on a given domain $A$.
The LoI ranks these equivalence relations based on the information they reveal about the underlying domain. Given two equivalence relations $P, Q \in \mathcal{L}(A)$, this ordering can be defined as follows:
\begin{align*}
	P \sqsubseteq Q \rightarrow \forall a,a' \in A \ \ (a \ Q \ a' \Rightarrow a \ P \ a')
\end{align*} 

For any set $S \subseteq \mathcal{L}(A)$, the least upper bound of $S$ is the equivalence relation $R$ defined as:
\begin{align*}
	\forall x,y \in A \ (x \ R \ y \leftrightarrow \forall P \in S. \ x \ P \ y).
\end{align*}

Formally, $LoI(A) = \tuple{\mathcal{L}(A), \sqsubseteq, \bigsqcup}$ denotes  the LoI on do\-main $A$, with ordering relation $\sqsubseteq$ and join $\bigsqcup$. %
The top element $\top$ in the lattice is the most precise equivalence relation \texttt{id} such that $\texttt{id} = \{ (x,y) \mid x, y \in A \wedge x = y\}$, and the bottom element $\bot$ is the least precise equivalence relation $\texttt{all} = \{ (x,y) \mid x, y \in A\}$. %

The join of any two equivalence relations $P\sqcup Q$ , being their least upper bound, is the \emph{least} informative equivalence relation that is at least as informative as either of $P$ and $Q$ (i.e. is an \emph{upper bound} on both), and thus represents the information that is conveyed from learning both $P$ \emph{and} $Q$. We refer to this as the \emph{conjunction} of the information in $P$ and $Q$.

\subsection{Quantale of Information}\label{sec:QoI}
The LoI captures the conjunction of any two information sources $P$ and $Q$ as the join of their respective equivalence relations.
However, it does not offer an operator that would yield a representation of their \emph{disjunction}, that is, the information that can be obtained from having access to one of them, but not both. In fact, the disjunction can not in general be represented as a single equivalence relation, and thus an element of the LoI, at all.
To address this limitation, Hunt and Sands~\cite{hunt2021quantale} propose a generalization of the LoI called the \emph{Quantale of Information} (QoI).
A quantale is a complete lattice with an additional binary ``tensor'' operator $\otimes$. In the QoI, the tensor is used to represent conjunction, while the lattice join represents \emph{dis}junction.

The core idea behind the quantale structure is to interpret the disjunction $P_1\vee\ldots\vee P_n$ of several knowledge relations as describing all knowledge relations $R$ in which the knowledge always comes from one of the $P_i$. More concretely, in any possible state $a\in A$, the agent's knowledge $[a]_R$ should equal its knowledge in the same state in one of the disjuncts, $[a]_{P_i}$. \emph{Which} disjunct it is may depend on the state, so the agent may have knowledge from $P_i$ in the state $a$ but knowledge from $P_j$ in some other state $a'$.
Relations $R$ that satisfy this condition are called \emph{tilings}, based on a picture of covering (since every state needs to be in some equivalence class) the space of possible states $A$ with knowledge sets drawn from any of the disjuncts. Following Hunt and Sands, we define the set of all tilings
\begin{align*}
	\mix(\P) = \{ R \in LoI(A) \mid x \in [R] \Rightarrow (\exists P \in \P . \, x \in [P]) \},
\end{align*}
where $\P$ is a set of equivalence relations.

We would like to think of a relation $R'$ as describing no more knowledge than a disjunction $\bigvee\P$ if it's bounded above by \emph{some} $R\in \mix(\P)$ in the LoI, and more generally define the quantale ordering $\S \sqsubseteq \T$ for $\S, \T \subseteq \mathcal{L}(A)$ as $\forall S \in \S, \ \exists T \in \T. \ S \sqsubseteq T$. 
The resulting relation is not antisymmetric on general sets of relations or even $\mix$es of general sets, reflecting the circumstance that there may be multiple $\mix$es representing the same knowledge.
As it is standard in lattice theory~\cite{davey2002introduction}, we use the downwards closure operator $\Downarrow$ to obtain canonical representations of the order cycles of $\sqsubseteq$ and hence construct a partial order.
\begin{align*}
	{\Downarrow} \P = \{ Q \in LoI(A) \mid Q \sqsubseteq \P \}
\end{align*}
The \emph{tiling closure} of a set of equivalence relations $\P$,
\begin{align*}
	\tc(\P) = {\Downarrow} \mix(\P),
\end{align*}
then canonically represents the knowledge permitted by the disjunction $\bigvee \P$. The set $\tc(\P)$ can still be interpreted as a list of possible equivalence relations, now including any equivalence relation that does not reveal more information than the disjunction.

We then take the elements of the QoI on a state set $A$ to be all tiling closures of subsets of $A$, with the ordering $\sqsubseteq$ being set inclusion.
For the tensor $\P \otimes \Q = \tc(\{ P \sqcup Q \mid P \in \P, Q \in \Q \})$, we rely on the join operator of the LoI $\sqcup$ to calculate the least upper bound of any possible pair of equivalence relations in $\P$ and $\Q$ and then canonicalise the result. Since the sets are interpreted disjunctively, the join $\bigvee_i \P_i$ can simply be defined as $\tc(\bigcup_i \P_i)$.

\begin{example}
    Program~\ref{prog:mix_example} operates on a secret integer $x$ between -2 and 3,
    outputting to user $u$ whether it is greater than zero, and \emph{either} (if it isn't) whether it is
    even, \emph{or} (if it is) whether it equals 0 or 1 (by dividing by 2, rounding down and testing for 0).
    We expect the information released by the program ($\sim_{\mathrm{\prg}}$ in Fig.~\ref{fig:mix_example}) to be bounded by the disjunction of the knowledge relations capturing the two possible branches (resp. $Q$, $P$).
\begin{figure}
	\begin{lstlisting}[label=prog:mix_example]
if (x <= 0) then
	out(-1  ,u);
	out(x mod 2 == 0, u);
else
	out(1, u);
	out(x div 2 == 0, u);
	\end{lstlisting}
\end{figure}
	
    This could not be accurately expressed with LoI operations, since $Q$, $P$ and $\sim_{\mathrm{\prg}}$ are all incomparable, but the join of $Q$ and $P$ (as the only available nontrivial way of combining them) is equal to $\top$ and so would equally bound a program that directly releases $x$.
	However, $\sim_{\mathrm{\prg}}$ can be tiled with equivalence classes from $Q$ and $P$, and we in fact have $\mix(\{Q,P\})=\{Q,P,R,\sim_{\mathrm{\prg}}\}$.
So in the QoI, $\tc(\{\sim_{\mathrm{\prg}}\}) \sqsubseteq \tc(\{Q,P\})$, and hence ${\sim_{\mathrm{\prg}}} \sqsubseteq Q \vee P$.
\def\cll#1{>{\arrayrulecolor{#1}} - >{\arrayrulecolor{black}}}
\def\hhhline#1{\expandafter\hhline\expandonce{#1}}
	\begin{figure}
		\centering
		\begin{tikzpicture}
			\setlength\arrayrulewidth{0.7pt}
			\begin{scope}[xscale=1.65]
				\node (all) [inner sep=0pt] at (0,0) {
					\begin{tabular}{|c c|}
						\hline
						-2 & -1 \\
						0 & 1 \\
						2 & 3 \\
						\hline
					\end{tabular}
				};
				\node[above] at (all.north) {$\mathrm{all}$};
				
				\node (Q) [inner sep=0pt] at (1,0) {
					\begin{tabular}{|c|c|}
						\hline
						\rowcolor{blueC!70} -2 & -1 \\
                        \hhline{>{\arrayrulecolor{blueC!70}} - >{\arrayrulecolor{black}}|-}
                        \rowcolor{blueC!70} 0 & 1 \\
                        \hhline{-|>{\arrayrulecolor{blueC!70}} - >{\arrayrulecolor{black}}}
                        \rowcolor{blueC!70}
						2 & 3 \\
						\hline
					\end{tabular}
				};
				\node[above] at (Q.north) {$Q$};
				
				\node (P) [inner sep=0pt] at (2,0) {
					\begin{tabular}{|c|c|}
						\hline
						\rowcolor{redC!20}
						\multicolumn{1}{|c}{-2} & -1 \\
						\hline
						\rowcolor{redC!20} 0 & 1 \\
						\hline
						\rowcolor{redC!20} \multicolumn{1}{|c}{2} & 3 \\
						\hline
					\end{tabular}
				};
				\node[above] at (P.north) {$P$};
				
				\node (E) [inner sep=0pt] at (3,0) {
					\begin{tabular}{|c|c|}
						\hline
						\cellcolor{blueC!70}{-2} & \cellcolor{blueC!70}{-1} \\
                        \hhline{>{\arrayrulecolor{blueC!70}} - >{\arrayrulecolor{black}}|-}
						\cellcolor{blueC!70}{0} & \cellcolor{redC!20}{1} \\
						\hline
						\multicolumn{1}{|c}{\cellcolor{redC!20}{2}} & \cellcolor{redC!20}{3} \\
						\hline
					\end{tabular}
				};
				\node[above] at (E.north) {$\sim_{\mathrm{\prg}}$};
				
				\node (R) [inner sep=0pt] at (4,0) {
					\begin{tabular}{|c|c|}
						\hline
                        \rowcolor{redC!20}
						\multicolumn{1}{|c}{-2} & -1 \\
						\hline
						\cellcolor{redC!20}0 & \cellcolor{blueC!70}1 \\
                        \hhline{-|>{\arrayrulecolor{blueC!70}} - >{\arrayrulecolor{black}}}
						\cellcolor{blueC!70}2 & \cellcolor{blueC!70}3 \\
						\hline
					\end{tabular}
				};
				\node[above] at (R.north) {$R$};
			\end{scope}
		\end{tikzpicture}
		\caption{Some equivalence relations on $\{-2,-1 ,0, 1 , 2 , 3 \}$}
		\label{fig:mix_example}
	\end{figure}
\end{example}

\section{Information Ordering in Databases}\label{sec:intro_to_databases}
Our goal is to introduce our semantic model for the information revealed by database queries,
the \emph{Determinacy Lattice}, and its extension to disjunctive dependencies, the \emph{Determinacy Quantale}.
To this end, we first review a standard formalism for reasoning about databases that we will employ.

\subsection{A Primer on Relational Database Models}\label{subsec:intro_to_databases}
We use the relational model to formally define databases \cite{abiteboul1995foundations}. 
In this model, we distinguish between the database schema $D$, which specifies the structure of the database, and the database state $db$, which specifies its actual content.

A database schema $D$ is a (nonempty) finite set of relation schemas $\tb$, written as $D = \{\tb_1, ..., \tb_n\}$.
A relation schema (table) $\tb$ is defined as a set of attributes paired with a set of constraints, where an attribute is a name paired with a domain.
The number of attributes in $\tb$ (written as $|\tb|$) is referred to as its arity. 
A tuple is a set of data representing a single record within a relation schema. Each tuple contains values for each attribute defined in the relation schema.

A \emph{database state} $db$ is a snapshot of the database schema $D$ at a particular point in time. It represents the actual data stored in the database, consisting of a collection of tables and their respective tuples. We write $\evalTable{\tb}{db}$ to represent the tuples of table $\tb$ under database state $db$.

We write $\mathrm{states}(D)$ to denote the set of all database states of $D$.
A database configuration is $\tuple{D, \Gamma}$ where $D$ is the database schema and $\Gamma$ is a set of integrity constraints. 
We denote $\Omega_D = \{ db \mid db \in \mathrm{states}(D) \ \wedge \vdash db : \Gamma \}$ where $\vdash$ is an appropriate notion of constraint $\Gamma$ being satisfied.
An integrity constraint is an assertion about a database that must be satisfied for a database state to be considered valid. Various classes of integrity constraints exist, for instance functional dependencies which capture primary-key constraints, and inclusion dependencies which are used in foreign-key constraints~\cite{abiteboul1995foundations}.

\tightpar{Relational calculus}
We rely on the Domain Relational Calculus (DRC) for our query language.  
In the DRC, a (non-boolean) query $q$ over a database schema $D$ has the form $\{\overline{x} \mid \phi\}$, where $\overline{x}$ is a sequence of variables, $\phi$ is a first order formula over $D$, and the free variables of $\phi$ are those in $\overline{x}$. The \emph{evaluation} of a query $q$, denoted by $\eval{q}{db}$, is the set of tuples that satisfy the formula $\phi$ with respect to $db$. A \emph{boolean query} is written as $\{ \ \mid{\phi}\}$, and its evaluation $\eval{q}{db}$ is defined to be the boolean value $\mathsf{true}$ if and only if some tuple in $db$ satisfies $\phi$.
We use $\queries$ to indicate the universe of all possible queries.

The domain relational calculus employed here follows the standard convention, and we refer the reader to the relevant literature for a more comprehensive description of DRC~\cite{abiteboul1995foundations}.

\begin{figure}[t]
	\centering
	\begin{align*}
		\mathrm{emp} : \ &\begin{tabular}{ |c|c|c| } 
			\hline
			\texttt{\underline{\textbf{n}}ame} & \texttt{\underline{\textbf{r}}ole} & \texttt{\underline{\textbf{s}}alary} \ \ \\ 
			\hline
		\end{tabular} \\ 
		\mathrm{mng} : \ &\begin{tabular}{ |c|c| } 
			\hline
			\texttt{\underline{\textbf{d}}ivision} & \texttt{\underline{\textbf{m}}anager} \\
			\hline
		\end{tabular}
	\end{align*}
	\caption{Database schema for employees and managers}
	\label{fig:relations_emp_man}
\end{figure}

\begin{example}
The database schema in Fig.~\ref{fig:relations_emp_man} contains relations for employees $\mathrm{emp}$ and managers $\mathrm{mng}$. A query returning the set of tuples containing the division names and the salary of the managers of each division can be written as:
\begin{align*}
	\{ (d, s) \mid \ &\exists n,r. \ \mathrm{emp}(n, r, s)  \wedge \exists m. \  \mathrm{mng}(d, m)  \wedge  n = m\}.
\end{align*}
\end{example}

\tightpar{Views}
In DRC, a database view is a relation defined by the result of a non-boolean query. Database views act as virtual tables and, as we will see, are useful when defining security policies. Formally, a view $\vw$ defined over database schema $D$ is a tuple $\tuple{id, q}$, where $id$ is the view identifier and $q$ is the non-boolean query over schema $D$ defining the view. The query $q$ may refer to other views, but we assume that views do not have cyclic dependencies.

The materialization of a view $\vw$ in a database state $db$ is the evaluation of its defining query $q$ in that state, \ie $\eval{q}{db}$. We use $\vw.q$ to refer to the defining query of view $\vw$. We extend relational calculus in the standard way to work with views~\cite{guarnieri2016strong}.

\subsection{Determinacy Lattice}\label{sec:determinacy_lattice}
Given query sets $Q, Q' \in \powerset{\queries}$, query determinacy~\cite{nash2010views} captures whether results of the queries in $Q$ are always sufficient to determine the result of the queries in $Q'$. %
\begin{definition}\label{def:query_determinacy}
	$Q$ determines $Q'$ (denoted by $Q \twoheadrightarrow Q'$) iff for all database states $db_1$, $db_2$, if $\eval{q}{db_1}$ = $\eval{q}{db_2}$ for all $q \in Q$, then $\eval{q'}{db_1}$ = $\eval{q'}{db_2}$ for all $q' \in Q'$.
\end{definition} 

Intuitively, $Q \twoheadrightarrow Q'$ means that pairs of databases for which all queries in $Q$ return the same result also give the same result under any query in $Q'$. 
This is in fact equivalent to the initial gloss that the results of queries in $Q'$ can be computed from the results of queries in $Q$, as we show in detail in {\iffull Appendix~\ref{sec:app:proof:determinacy_def_equiv}. \fi} 

Query determinacy allows us to define an ordering on sets of queries based on the information they reveal. We call this ordering \emph{determinacy order}, denote it by $\preceq$, and define it as $\forall Q, Q' \in \powerset{\queries}$, $Q \preceq Q'$ iff $Q' \twoheadrightarrow Q$. 

\begin{example}
	Consider queries $q_1 = \{(n,r) \mid \exists s. \ \mathrm{emp}(n,r,s)\}$ and $q_2 = \{(r) \mid \exists n,s. \ \mathrm{emp}(n,r,s)\}$
	defined on the relations of Fig.~\ref{fig:relations_emp_man}.
	Query $q_1$ discloses the $\mathrm{name}$ and the $\mathrm{role}$ of the employees while $q_2$ only returns their $\mathrm{role}$. Intuitively, $q_1$ reveals more information than $q_2$, which means $q_2 \preceq q_1$.
\end{example}

This definition of determinacy order is a preorder (reflexive and transitive), but not necessarily a partial order, as it is not anti-symmetric. In other words, $q_1 \preceq q_2$ and $q_2 \preceq q_1$ does not necessarily mean that $q_1 = q_2$. 
As in Section~\ref{sec:LoI}, this essentially means that query sets are not canonical representations of the information revealed by them. To rectify this, we form the closure $\clEmpty$ under the determinacy order, so the determinacy order becomes set inclusion.
Intuitively, $\cl{Q}$ will contain all the queries in $\queries$ whose answers can be inferred by the set of queries $Q$. Formally,  $\cl{Q}$ is defined as:
\begin{align*}
	\cl{Q} = \{ q \in \queries \mid \{q\} \preceq Q \}
\end{align*}

Using the definitions of determinacy order and closure $\clEmpty$, we can then define the Determinacy Lattice as follows:
\begin{definition}\label{def:determinacy_lattice}
	Given a universe of queries $\queries$, the Determinacy Lattice $DL(\queries)$ is a complete lattice $\tuple{\mathcal{L}, \sqsubseteq, \bigsqcup, \bot, \top}$ such that:
	\begin{itemize}
		\item $\mathcal{L} = \{ \cl{Q} \mid Q \subseteq \queries \}$
		\item $\cl{Q_1} \sqsubseteq \cl{Q_2}$ iff $Q_1 \preceq Q_2$
		\item $\bigsqcup_i \cl{Q_i}  = \cl{\bigcup_i Q_i}$
		\item $\bot = \cl{\varnothing}$, $\top = \cl{\queries}$,
	\end{itemize}
where $\preceq$ is the determinacy order on $\queries$.
\end{definition}

\tightpar{Disclosure order and information flow properties}
Our definition of the Determinacy Lattice is similar to the definition of the Disclosure Lattice introduced by Bender et al.~\cite{bender2013fine}. A Disclosure Lattice is a lattice built upon a disclosure order, which is a partial order on sets of queries satisfying additional conditions that are expected of an ordering according to the amount of information disclosed by each set of queries. Bender et al.~\cite{bender2013fine} define the disclosure order as follows: 
\begin{definition}\label{def:disclosure_order}
	Given a universe of queries $\queries$, a disclosure order $\preceq$ is a preorder on $\powerset{\queries}$ that satisfies the following properties:
	\begin{enumerate}
		\item For all $Q_1, Q_2 \in \powerset{\queries}$, if $Q_1 \subseteq Q_2$ then $Q_1 \preceq Q_2$
		\item If $\P \subseteq \powerset{\queries}$ and $\forall P \in \P, \ P \preceq Q$ then $\bigcup \P \preceq Q$
	\end{enumerate}
\end{definition}
The first property in this definition ensures that adding new elements to a set of queries only increases the amount of disclosed information and the second property allows us to derive a meaningful upper bound on the information disclosure.

The intended use of disclosure order was to order sets of queries based on the amount of information they reveal about the underlying database. However, we make the observation that this definition is not specific enough to characterize information disclosure in the information flow sense. For example, consider query containment~\cite{abiteboul1995foundations}, defined as:
\begin{definition}\label{def:query_containment}
	Given queries $q_1, q_2 \in \queries$, we say that $q_1$ is contained in $q_2$, denoted by $q_1 \subseteq q_2$, if for every database states $db \in \Omega_D$, we have $\eval{q_1}{db} \subseteq \eval{q_2}{db}$.
\end{definition}

Query containment satisfies all of the requirements of a disclosure order (Def.~\ref{def:disclosure_order}), but it is not enough to guarantee security. To illustrate this, consider a database with a single table $\tb$ given in Fig.~\ref{fig:table_exp_disclosure_dis}. 
\begin{figure}[H]
	\centering
	\begin{tabular}{ |c|} 
		\hline
		$\vl$ \\
		\hline
		\hline
		\texttt{$0$} \\
		\hline
		\texttt{$1$} \\
		\hline
		\texttt{$100+s$} \\
		\hline
	\end{tabular}
	\caption{Table $t$}
	\label{fig:table_exp_disclosure_dis}
\end{figure}

Table $\tb$ has a single column $\vl$, and contains values $0$, $1$, and $100+s$, where $s$ is a secret value that can be either $0$ or $1$. We thus consider two possible instances of this database, one where $\tb$ contains values $0$, $1$, and $100$ and another where it contains $0$, $1$, and $101$.
Now, consider the following queries:
\begin{align*}
	q_1 : \{ (\vl_1) \mid  \exists \vl_2. \ \tb1(\vl_1) \wedge \tb_2(\vl_2) &\wedge \vl_1 < 100 \} \\
	q_2 : \{ (\vl_1) \mid  \exists \vl_2. \ \tb1(\vl_1) \wedge \tb_2(\vl_2) &\wedge \vl_1 < 100 \\ &\wedge \vl_1 = \vl_2 - 100 \}
\end{align*}
\noindent
where $\tb_1$ and $\tb_2$ are just logical copies of table $\tb$. It is common practice to make logical copies of relation and use them in queries with
self-joins~\cite{wang2022conjunctive}. The result of query $q_1$ is always $0$ and $1$. The result of query $q_1$ is $1$ if the secret $s$ is $1$ and $0$ if $s$ is $0$. As it is evident, for these queries, query containment holds and the result of query $q_2$ is contained in the results of $q_1$. However, an observer seeing the result of query $q_2$ can learn the value of secret $s$.

This example illustrates that query containment (a disclosure order) is not sufficient to guarantee the confidentiality of the secret $s$ in an information flow setting. To ensure information flow security, we require a stronger condition, such as the notion of query determinacy order (Def.~\ref{def:query_determinacy}) that we chose to rely on in this paper.

\tightpar{Relation between the DL and the LoI} There exists a close relationship between the DL and the LoI. Specifically, a query $q$ defined over a database schema $D$ induces an equivalence relation $\eq{q}$ on database states $db$. We can formally define this equivalence relation as:
\begin{align*}
	\eq{q} = \{ (db_1 , db_2) \mid db_1, db_2 \in \Omega_D \wedge \eval{q}{db_1} = \eval{q}{db_2} \}
\end{align*}

We write $[\eq{q}]$ to denote the set of all equivalence classes induced by $q$. Given an equivalence relation $\eq{q}$ on set $\Omega_D$ and $db \in \Omega_D$, $[db]_{\eq{q}}$ denotes the equivalence class induced by $\eq{q}$ to which the database state $db$ belongs. 
We further lift this definition to sets of queries $Q = \{q_1,q_2,...,q_n\}$:
\begin{align*}
	\eq{Q} = \{ (db_1 , db_2) \mid db_1, db_2 \in \Omega_D &\bigwedge_{1 \leq i \leq n} \eval{q_i}{db_1} = \eval{q_i}{db_2} \}
\end{align*}

This interpretation of database queries as equivalence relations provides a direct connection between the DL and the LoI, where the lattice elements correspond to $\eq{Q}$, the ordering $\sqsubseteq$ to the determinacy order $\preceq$, and join and meet follow the definitions of the DL.

\begin{restatable}{globallemma}{determinacyLoILattice}\label{lemma:DL_imply_LoI}
	For all $\queries$, there is a complete lattice homomorphism from the Determinacy Lattice $DL(\queries)$ to the Lattice of In\-formation defined on $\{\eq{Q} \mid Q \in DL(\queries)\}$.
\end{restatable}
{\iffull We prove this Lemma in Appendix~\ref{sec:app:proof:determinacy_loi_lattice}.\fi}
To the extent that we believe $\eq{Q}$ to accurately represent the information conveyed by the queries in $Q$, this lemma implies that joins and order comparisons can be performed in the DL without explicit reference to the LoI.

\subsection{Determinacy Quantale}\label{sec:determinacy_quantale}
We introduce a generalization of the Determinacy Lattice, called the \emph{Determinacy Quantale} (DQ),  to represent disjunctive dependencies. 
Our definition of the DQ is intended as a counterpart to the QoI~\cite{hunt2021quantale}, analogously to how the DL corresponds to the LoI. To achieve this, we define a query-set counterpart of the tiling closure operator to capture the disjunction of sets of queries. Since \emph{sets} of queries correspond to LoI elements (equivalence relations), disjunctive QoI elements (sets of equivalence relations) will be represented as \emph{sets of sets} of queries. Each set of queries in the outer set represents a possible combination of queries that does not reveal more information than is allowed by the disjunction.

Analogously to the QoI, the tiling closure of a set of sets of queries is defined by forming the downward closure under $\sqsubseteq$ (from the DL) of their \emph{mix}.
The query-set equivalent of the \emph{mix} operator is defined on a set of sets of queries $\Q = \{Q_1,...,Q_n\}$ such that $Q_i \in DL(\queries)$ for $\ i=1,...n$ as follows:
\begin{align*}
	\nonumber
	\mix(\Q) = \{ P \in DL(\queries) \mid x \in [\eq{P}] \Rightarrow (\exists Q \in \Q. x \in [\eq{Q}]) \}
\end{align*}
where $[\eq{Q}]$ denotes the equivalence classes of $Q$ as defined previously.
We then define the tiling closure for a set $\Q$ of elements of the DL as 
$	\tc(\Q) = {\Downarrow} \mix(\Q)$.

We then formally define the Determinacy Quantale $DQ(\queries)$ as follows.
\begin{definition}\label{def:determinacy_quantale}
Given a universe of queries $\queries$, 
	let $DL(\queries)$ be the Determinacy Lattice defined on $\queries$. The Determinacy Quantale $DQ(\queries)$ is the quantale $\tuple{\mathcal{I}, \sqsubseteq, \bigvee, \otimes, 1}$, with:
	\begin{itemize}
		\item $\mathcal{I} = \{\tc(\Q) \mid \Q \subseteq DL(\queries) \}$
		\item $\bigvee_i \P_i = \tc(\bigcup_i \P_i)$
		\item $\P \otimes \Q = \tc\Big( \bigcup_{P \in \P, Q \in \Q} (P \sqcup Q)\Big)$
		\item $\sqsubseteq = \subseteq $
		\item $\top = DL(\queries)$, $\bot = \varnothing$, $1 = \varnothing$,
	\end{itemize}
 where $\P, \Q \subseteq DL(\queries)$.
\end{definition}

In {\iffull Appendix~\ref{sec:app:proof:determinacyQuantale} \fi} we show that Def.~\ref{def:determinacy_quantale} satisfies the usual quantale axioms~\cite{hunt2021quantale}.
As with the DL and LoI, the DQ embeds into a QoI by a quantale homomorphism. This QoI is defined
on sets of equivalence relations derived from sets of sets of queries by the following map:
\begin{definition}\label{def:set_query_to_set_eq}
	Given a set of sets of queries $\Q$, $$\EQ{\Q} = \{ \eq{Q} \mid Q \in \Q\}.$$
\end{definition}
We can then formally state the relationship between the DQ and this quantale as follows.
\begin{restatable}{globallemma}{DQimplyQoI}\label{lemma:DQ_imply_QoI}
	For all $\queries$, there is a quantale homomorphism from the Determinacy Quantale $DQ(\queries)$ to the Quantale of Information defined on $\{\EQ{\Q} \mid \Q \subseteq DL(\queries)\}$.
\end{restatable}
{\iffull The proof of Lemma~\ref{lemma:DQ_imply_QoI} is presented in Appendix~\ref{sec:app:proof:DQ_imply_QoI}.\fi}

\begin{example}
To illustrate the Determinacy Quantale in practice, consider Program~\ref{prog:DQ_example}, which issues either query 
$q1 = \{(r,\vl) \mid \exists s,n. \, \mathrm{emp}(n,r,s) \wedge r=\mathrm{Intern} \wedge \vl = s\}$ 
or 
$q2 = \{(r,\vl) \mid \exists s,n. \, \mathrm{emp}(n,r,s) \wedge r=\mathrm{CEO} \wedge \vl = n )\}$ 
to the database. Query q1 returns the $\mathrm{role}$ and $\mathrm{salary}$ columns of the entry in table $\mathrm{emp}$ if the role of that entry is $\mathrm{Intern}$. Similarly, query q2 returns the $\mathrm{role}$ and $\mathrm{name}$ columns if the role of the entry in $\mathrm{emp}$ is $\mathrm{CEO}$.

\begin{figure}
\begin{lstlisting}[label=prog:DQ_example]
if (y > 0) then
	x $\leftarrow$ q1
else
	x $\leftarrow$ q2
out(x, u);
\end{lstlisting}
\end{figure}

Consider a policy defined on queries $\vw1 = \{ (r,n) \mid \exists s. \, \mathrm{emp}(n,r,s) \}$ and $\vw2 = \{ (r,s) \mid \exists n. \, \mathrm{emp}(n,r,s) \}$. %
\vw1 and \vw2, which respectively project on the $\mathrm{name}$ and $\mathrm{role}$, and the $\mathrm{role}$ and $\mathrm{salary}$ columns of $\mathrm{emp}$, are used in defining the disjunctive security policy $\vw1 \vee \vw2$. 

For this example, we assume a database that has only one row in the $\mathrm{emp}$ table, and we also limit the domain of possible roles to $\{\mathrm{CEO}, \mathrm{Intern}\}$. These limitations are necessary in order to have a finite representation of the potential query sets and enables us to effectively depict the sets produced by the $\mix$ and $\tc$ operators.

Program~\ref{prog:DQ_example} depicts a disjunction that -- ignoring variable \texttt{y} -- depends either on $q1$ or $q2$ (\ie $q1 \vee q2$), which on the DQ can be represented as a point $\tc(\cl{\{q1\}}) \vee \tc(\cl{\{q2\}})$. Similarly, the policy $\vw1 \vee \vw2$ can be represented on the DQ by $\tc(\cl{\{\vw1\}}) \vee \tc(\cl{\{\vw2\}})$. 

Illustrating this point requires calculating the $\mix$ set of \vw1 and \vw2, which includes all sets of queries whose equivalence relation can be constructed from the equivalence classes of $\eq{\cl{\{\vw1\}}}$ and $\eq{\cl{\{\vw2\}}}$. 
Unfortunately, for any sufficiently rich query language, our definition of $\mix$ inevitably yields an infinite set, as infinitely many queries that are ``morally equivalent'' or even the same up to renaming variables represent the same knowledge set. 
To compactly represent such infinite sets, we will pick just one representative, and define
\begin{align*}
	hc(\Q) = \{ Q' \mid \exists Q \in \Q . \ \eq{Q} = \eq{Q'} \}
\end{align*}
as a closure operator that adds all equivalent queries.
Then $\mix\big( \{\clSet{\vw1}, \clSet{\vw2}\} \big)$ will be the set $hc(\{\clSet{\vw1}$, $\clSet{\vw2}$, $\clSet{p1}$, $ \clSet{p2}\})$, where
$p1 = \{(r,\vl) \mid \big(\exists s,n. \, \mathrm{emp}(n,r,s) \wedge r=\mathrm{Intern} \wedge \vl = s \big) \vee \big(\exists s,n. \, \mathrm{emp}(n,r,s) \wedge r=\mathrm{CEO} \wedge \vl = n \big)\}$
and 
$p2 = \{(r,\vl) \mid \big(\exists s,n. \, \mathrm{emp}(n,r,s) \wedge r=\mathrm{CEO} \wedge \vl = s\big) \vee \big(\exists s,n. \, \mathrm{emp}(n,r,s) \wedge r=\mathrm{Intern} \wedge \vl = n \big) \}$.

Therefore, we can depict the policy as the point ${\Downarrow} (hc(\{\clSet{\vw1}$, $\clSet{\vw2}$, $\clSet{p1}$, $ \clSet{p2}\}))$ on the DQ.
Similarly, the DQ point of the Program~\ref{prog:DQ_example} (\ie $\tc(\cl{\{q1\}}) \vee \tc(\cl{\{q2\}})$), can also be depicted by the point ${\Downarrow}hc(\{\clSet{p1}\})$ on the DQ. 
We illustrate the part of the DQ which includes these points in Fig.~\ref{fig:DQ_example}, and as it is evident from the figure, conclude that Program~\ref{prog:DQ_example} is inline with the policy.

\begin{figure}[H]
	\centering
	\begin{tikzpicture}
		\begin{scope}[xscale=2, yscale=1.5]
			\node (vw1) at (-1.9,0) {$\tc(\cl{\{\vw1\}})$};
			\node (vw2) at (-0.9,0) {$\tc(\cl{\{\vw2\}})$};
			\node (q1) at (0.2,0) {$\tc(\cl{\{q1\}})$};
			\node (q2) at (1.2,0) {$\tc(\cl{\{q2\}})$};
			
			\node (q1vq2) at (0.7,1) {${\Downarrow}hc(\{\clSet{p1}\})$};
			\node (vw1vvw2) at (-0.7,2) {${\Downarrow}hc(\{\clSet{\vw1}, \clSet{\vw2}, \clSet{p1}, \clSet{p2}\})$};
			
			\draw[-] (q1)--(q1vq2);
			\draw[-] (q2)--(q1vq2);
			\draw[-] (vw1)--(vw1vvw2);
			\draw[-] (vw2)--(vw1vvw2);
			
			\draw[-] (q1vq2)--(vw1vvw2);
		\end{scope}
	\end{tikzpicture}
	\caption{A portion of the DQ for queries q1, q2, \vw1, \vw2}
	\label{fig:DQ_example}
\end{figure}

\end{example}

\section{Security Framework}\label{sec:system_model}
Drawing on the quantale model of dependencies for programs and databases, we develop an extensional condition that defines security for programs that interact with databases and support disjunctive security policies. We will later use the security condition to prove soundness of enforcement mechanisms  in Section \ref{sec:enforcement}.  
Specifically, we formalize the syntax and semantics of a simple imperative language with database queries. Programs read the input from the database via queries, while users receive the output through predefined output channels. We define (disjunctive) security policies as views over the database and interpret them end-to-end.   We then use this model to define a knowledge-based security condition for our setting.

\subsection{Language}\label{sec:language}
\tightpar{Syntax}
The syntax for the commands of our language as depicted in Fig.~\ref{fig:syntax-commands}, primarily consists of standard commands such as assignment, conditionals, and loops.
The command $\texttt{out}(e,u)$  outputs the result of evaluating expression $e$ to user $u \in \users$.
The command $x \leftarrow q$ issues the query $q$ to the database and stores the result  in variable $x$. For modeling the queries, we rely on conjunctive queries with comparison introduced in Section~\ref{sec:CQC}.

Expressions $e$ can be variables $x \in \mathrm{Vars}$, values (integers) $n \in \mathrm{Val}$, binary operations $e_1 \oplus e_2$, single tuples $\tp \in \mathrm{Val}$, and set of tuples $\overline{\tp} \in \mathrm{Val}$. For simplicity, we do not provide de-constructors for database tuples.

\begin{figure}
	\centering
	{
		\setstretch{1.5}
		$
		\begin{array}{ll}
			c :=&\texttt{skip} \sep \texttt{if} \ e \ \texttt{then} \ c_1 \ \texttt{else} \ c_2 \sep\\
			&x \leftarrow q \sep x := e \sep c_1;c_2 \sep\\
			&\texttt{while} \ e \ \texttt{do} \ c \sep \texttt{out}(e,u) \\
		\end{array}
		$
	}
	\caption{Language syntax}
	\label{fig:syntax-commands}
\end{figure}

\tightpar{Semantics}
As discussed in  Section~\ref{sec:determinacy_quantale}, a database state (or simply state) $db \in \Omega_D$ is defined with respect to a  schema  $D$ and  a finite set of integrity constraints.
A configuration $\tuple{c, m , db}$ consists of a command $c$, a memory $m = \mathrm{Var} \rightarrow \mathrm{Val}$ mapping variables to values, and a state $db$.

The semantics of expressions is mostly standard and its rules are presented in Fig.~\ref{fig:op_sem_exp_full}.
We use judgments of the form $\tuple{e, m, db}  \downarrow \vl$ to denote that an expression $e$ evaluates to value $\vl$ in memory $m$ and state $db$. For simplicity, we refrain from defining binary operations on tuples, unless the underlying database query is boolean.

We use judgments of the form $\tuple{c, m, db} \xrightarrow{\alpha} \tuple{c', m', db'}$ to denote that a configuration $\tuple{c, m, db}$ in one step evaluates to memory $m'$ and state $db'$ and (possibly) produces an observation $\alpha \in \mathrm{Obs}$; we write $\epsilon$ whenever a command produces no observation. We write $m[x \mapsto \vl]$ to denote a memory $m$ with variable $x$ assigned the value $\vl$. 

Fig.~\ref{fig:op_sem_cmds_full} provides the semantic rules for commands.
The query evaluation rule \textsc{QueryEval} is similar to assignment as it evaluates a query $q$ into state $db$ and stores the result in the variable $x$. 
We use the command $\texttt{out}(e,u)$ to produce an observation. Formally, an observation $\alpha \in \mathrm{Obs}$ is a tuple $\tuple{o, u}$, where $u \in \users$ is the identifier of the user observing the output and $o$ is the result of evaluating expression $e$, which is either a simple value or the result set of a non-boolean query. %

We write $\tuple{c, m, db} \mulStepEval{\tau}{u} \tuple{c', m', db'}$ to denote when $\tuple{c, m, db}$ takes one or more steps to reach configuration $\tuple{c', m', db'}$ while producing the trace (sequence of observations) $\tau \in \mathrm{Obs}^\ast $. We omit the final configuration whenever it is irrelevant and write $\tuple{c, m, db}  \mulStepEval{\tau}{u}$.

\begin{figure*}[h]
	{
		\setstretch{1.4}
		\footnotesize
		\centering
		$
		\inferrule*[before=\textsc{Int}]
		{
			\\
		}
		{
			\tuple{n, m, db}  \downarrow n
		}
		$
		\nextrule
		$
		\inferrule*[before=\textsc{Tuple}]
		{
			\\
		}
		{
			\tuple{\tp, m, db}  \downarrow \tp
		}
		$
		\nextrule
		$
		\inferrule*[before=\textsc{TupleSet}]
		{
			\\
		}
		{
			\tuple{\overline{\tp}, m, db}  \downarrow \overline{\tp}
		}
		$
		\nextrule
		$
		\inferrule*[before=\textsc{Var}]
		{
			\vl = m(x)
		}
		{
			\tuple{x, m, db}  \downarrow \vl
		}
		$
		
		\nextrule
		
		$
		\inferrule*[before=\textsc{Op}]
		{
			\tuple{e_1, m, db}  \downarrow n_1 \\
			\tuple{e_1, m, db}  \downarrow n_2 \\
			n = n_1 \oplus n_2
		}
		{
			\tuple{e_1 \oplus e_2, m, db}  \downarrow n
		}
		$
		
	}
	\caption{Semantic rules for expressions}
	\label{fig:op_sem_exp_full}
\end{figure*}
\begin{figure*}[h]
	{
		\setstretch{1.4}
		\footnotesize
		\centering
		$
		\inferrule*[before=\textsc{Skip}]
		{
			\\
		}
		{
			\tuple{\texttt{skip}, m, db} \xrightarrow{\epsilon} \tuple{\epsilon, m, db}
		}
		$
		\nextrule
		$
		\inferrule*[before=\textsc{Assign}]
		{
			\tuple{e, m, db}  \downarrow \vl \\
			m' = m[x \mapsto \vl]
		}
		{
			\tuple{x := e, m, db} \xrightarrow{\epsilon} \tuple{\epsilon, m', db}
		}
		$
		\nextrule
		$
		\inferrule*[before=\textsc{QueryEval}]
		{
			\vl = \eval{q}{db} \\
			m' = m[x \mapsto \vl]
		}
		{
			\tuple{x \leftarrow q, m, db} \xrightarrow{\epsilon} \tuple{\epsilon, m', db}
		}
		$
		
		\nextrule
		
		$
		\inferrule*[before=\textsc{IfTrue}]
		{
			\tuple{e, m, db} \downarrow n \\
			n \not= 0
		}
		{
			\tuple{\texttt{if} \ e \ \texttt{then} \ c_1 \ \texttt{else} \ c_2, m, db} \xrightarrow{\epsilon} \tuple{c_1, m, db}
		}
		$
		\nextrule
		$
		\inferrule*[before=\textsc{IfFalse}]
		{
			\tuple{e, m, db} \downarrow n \\
			n = 0
		}
		{
			\tuple{\texttt{if} \ e \ \texttt{then} \ c_1 \ \texttt{else} \ c_2, m, db} \xrightarrow{\epsilon} \tuple{c_2, m, db}
		}
		$
		
		\nextrule
		
		$
		\inferrule*[before=\textsc{WhileTrue}]
		{
			\tuple{e, m, db}  \downarrow n \\
			n \not= 0
		}
		{
			\tuple{\texttt{while} \ e \ \texttt{do} \ c, m, db} \xrightarrow{\epsilon} \tuple{c;\texttt{while} \ e \ \texttt{do} \ c, m, db}
		}
		$		
		\nextrule
		$
		\inferrule*[before=\textsc{WhileFalse}]
		{
			\tuple{e, m, db} \downarrow n \\
			n = 0
		}
		{
			\tuple{\texttt{while} \ e \ \texttt{do} \ c, m, db} \xrightarrow{\epsilon} \tuple{\epsilon, m, db}
		}
		$	
		
		\nextrule
		
		$
		\inferrule*[before=\textsc{Seq}]
		{
			\tuple{c_1, m, db} \xrightarrow{\alpha} \tuple{c_{1}', m', db'} \\
		}
		{
			\tuple{c_1;c_2, m, db} \xrightarrow{\alpha} \tuple{c_{1}';c_2, m', db'}
		}
		$
		\nextrule
		$
		\inferrule*[before=\textsc{SeqEmpty}]
		{
			\\
		}
		{
			\tuple{\epsilon;c, m, db} \xrightarrow{\epsilon} \tuple{c, m, db}
		}
		$
		\nextrule
		$
		\inferrule*[before=\textsc{Output}]
		{
			\tuple{e, m, db}  \downarrow \vl
		}
		{
			\tuple{\texttt{out}(e,u), m, db} \xrightarrow{\tuple{\vl, u}} \tuple{\epsilon, m, db}
		}
		$
		
	}
	\caption{Semantics rules for commands}
	\label{fig:op_sem_cmds_full}
\end{figure*}

\subsection{Security Model}\label{sec:security_model}
We now introduce our knowledge-based security model for disjunctive security policies. For simplicity, we denote the initial program memory by $\initM$ and assume it is fixed and public to all users, hence the only way to input sensitive information is through database queries. Users make observations  through output channels, hence their knowledge of the database is determined by what they can infer based on these observations. This model induces standard equivalence relations for database states and observation traces.

\tightpar{Database state equivalence}
Two states $db$ and $db'$ are equivalent with respect to a set of tables and views $\stv$, written as $db \approx_{\stv} db'$, iff all tables and views in $\stv$ have identical contents in $db$ and $db'$. Formally, states $db$ and $db'$ are equivalent with respect to $\stv$ iff for all view $\vw \in \stv, \ \eval{\vw.q}{db} = \eval{\vw.q}{db'}$ and for all table $\tb \in \stv, \ \evalTable{\tb}{db} = \evalTable{\tb}{db'}$.
A set of tables and views $\stv$ induces an equivalence relation, and for a state $db$, the equivalence class $[db]_{\stv}$ contains all states that are equivalent to $db$ with respect to $\stv$.

\tightpar{Trace equivalence}
We use trace projection to define trace equivalence. The projection of a trace $\tau$ for user $u$ written as $\proj{\tau}{u}$ is the sequence of all observations in $\tau$ that $u$ can observe. Traces $\tau_1$ and $\tau_2$ are equivalent with respect to user $u$, written as $\tau_1 \approx_{u} \tau_2$, iff the projection of one of them to $u$ is the prefix of the other, \ie \ $\proj{\tau_1}{u} \ \preceq \proj{\tau_2}{u}$ or $\proj{\tau_1}{u} \ \succeq \proj{\tau_2}{u}$. 

Equivalence of trace prefixes is a standard  technicality  needed to ignore leaks due to program's progress/termination~\cite{askarov2008termination}, and here  we adapt a definition of trace equivalence which does not differentiate between program divergence and termination~\cite{guarnieri2019information}. 

\tightpar{User knowledge}
When executing a program $\prg$, we assume memory is always initially in the all-zero state $m_0$. Thus, we can
view a program's execution for any user as a function from database $db$ to user-observable output traces, $\tau_{\prg,u}(db) = \proj{\tau}{u}$ when $\tuple{\prg, m_0, db} \mulStepEval{\tau}{u}$. This function induces an equivalence relation on databases, $\llbracket \prg\rrbracket_u = {\sim_{\tau_{\prg,u}}}$, which characterizes the knowledge of $db$ conveyed by the output of $\prg$ to $u$.

\tightpar{Security policy}\label{sec:security_policy_syntax}
A security policy is a list of user policies (written as $P_u$) for each user $u \in \users$. User policies are defined as views and table identifiers over a database schema, and determine what a user $u$ is allowed to observe. Fig.~\ref{fig:policy_grammer} presents the syntax of disjunctive policies for our model. 
They are defined as a set of sets in order to represent a disjunction of conjunctions of simpler policies. A conjunction $\con$ is  a set of view $\vw$ and table $\tb$ identifiers, and a disjunction $\dis$ is a set of conjunctions. 
For example, the policy $P_u$  for user $u$ who is allowed to see table $\tb_1$ and view $\vw_1$, or view $\vw_2$ but not both, is defined as  $P_u = \{ \{\tb_1, \vw_1\} , \{\vw_2\} \}$. 

The overall policy of the system, written as $P$, is the list of user policies. %
Per Def.~\ref{def:set_query_to_set_eq}, the policy $P_u$ can be represented semantically as an element $\llbracket P_u\rrbracket$ of the Quantale of Information.
Thus, we can formulate our security condition as the assertion that the knowledge of the database
that the execution of the program $\prg$ conveys to $u$ is bounded above by the disjunctive knowledge
allowed by the policy, $\llbracket P_u\rrbracket$.

\begin{definition}\label{def:security_condition}
The program $\prg$ is secure for the user $u$ and policy $P_u$ if $\llbracket \prg \rrbracket_u \sqsubseteq \llbracket P_u \rrbracket$.
\end{definition}

\begin{figure}
	\centering
	{
		\setstretch{1.5}
		\centering
		$
		\begin{array}{rcl}
				\con &:=& \{\vw\} \mid \{\tb\} \mid \con_1 \cup \con_2 \\
				\dis &:=& \{\con\} \mid \dis_1 \cup \dis_2 \\
				P_u &:=& \dis \\
			\end{array}
		$
	}
	\caption{Syntax of user policy}
	\label{fig:policy_grammer}
\end{figure}

\section{Enforcement of Disjunctive Policies}\label{sec:enforcement}
Having formulated the security condition, we would like to prove
that useful programs satisfy it. To this end, we introduce
a sound static enforcement mechanism, which imposes some
structural limitations on the policy and trades off some completeness
for the sake of efficiency and ease of analysis.

Fig.~\ref{fig:enforcement_highlevel} illustrates how our
mechanism functions at a high level. We assume as input a program and policy
in the format described in Fig.~\ref{fig:syntax-commands} and Fig.~\ref{fig:policy_grammer}
respectively. The program is then subjected to a static \emph{dependency analysis} (Section~\ref{sec:dep_analysis}),
which computes an overapproximate set of possible paths of control flow
through the program, along with the queries (dependencies) retrieved for each path, 
giving an element of the DQ, that is a (disjunctive) set of (conjunctive) sets of queries.
Per Fig.~\ref{fig:policy_grammer}, the policy is also already given in this format.

We would like to verify that the program dependencies are bounded by the policy in the DQ,
as by Lemma \ref{lemma:DQ_imply_QoI}, this entails the security condition (Def.~\ref{def:security_condition}) that the disjunctive information
that is revealed by the program is bounded above by the QoI interpretation of the policy.
However, checking DQ ordering on general queries may be computationally costly.
We therefore \emph{abstract} (Section~\ref{sec:symbolic_tupes}) both the policy and the path dependencies into a more tractable format (symbolic tuples), which again overapproximates the
information they can retrieve. 
To guarantee soundness, we require that the views in the policy are
such that this abstraction is lossless for them. 
Finally, as the \emph{security check} (Section~\ref{sec:enforcement_security_check}),
we compute a tractable comparison on sets of sets of symbolic tuples that 
can be shown to imply DQ ordering.

\begin{figure}
	\centering
	\begin{tikzpicture}
		\begin{scope}[xscale=2.2, yscale=1.5]
			\node[draw, ellipse, minimum width=1.5cm] (pol) at (0,1) {\scriptsize Policy};
			\node[draw, ellipse, minimum width=1.5cm] (prg) at (0,0) {\scriptsize Program};
			\node[draw, rounded corners] (da) at (1,0) {\scriptsize \begin{tabular}{c} Dependency \\ Analysis\end{tabular}};
			\node[draw, rounded corners] (qaPol) at (2,1) {\scriptsize \begin{tabular}{c} Query \\ Abstraction\end{tabular}};
			\node[draw, rounded corners] (qaPrg) at (2,0) {\scriptsize \begin{tabular}{c} Query \\ Abstraction\end{tabular}};
			\node[draw, circle] (cnd) at (3,0.5) {\scriptsize \begin{tabular}{c} Security \\ Check\end{tabular}};
			
			\draw[->] (prg)--(da);
			\draw[->] (pol)--(qaPol);
			\draw[->] (da)--(qaPrg);
			\draw[->] (qaPol)--(cnd);
			\draw[->] (qaPrg)--(cnd);
		\end{scope}
	\end{tikzpicture}
	\caption{Enforcement steps}
	\label{fig:enforcement_highlevel}
\end{figure}

\subsection{Conjunctive Queries}\label{sec:CQC}
While our theoretical definitions are based on the fully-general
domain relational calculus as a query language, to avoid complexity, our enforcement mechanism
will work with a restricted subset called \emph{conjunctive queries with comparisons} (CQCs).
This language is a subset of relational calculus that only employs conjunction ($\wedge$) and existential quantification ($\exists$) and omits disjunction ($\vee$), negation ($\neg$), and universal quantification ($\forall$). CQCs can model \texttt{SELECT}-\texttt{FROM}-\texttt{WHERE} portion of SQL, where there are only \texttt{AND} and comparisons in the \texttt{WHERE} clause.

Our language for (non-boolean) CQC $q$ over a database schema $D$ employs the standard notation~\cite{abiteboul1995foundations,wang2022conjunctive}, and has the form $\emph{heading} \leftarrow \emph{body}$:
\begin{align*}
	\mathrm{ans}(\overline{y}) \leftarrow R_1(\overline{x}_1), ..., R_n(\overline{x}_n), C_1, ..., C_m 
\end{align*}
where $R_1, ..., R_n$ are relations in $D$, and $\overline{x}_1, ..., \overline{x}_n$ are their variables. We use $\mathrm{Var}(q) = \overline{x}_1 \cup ... \cup \overline{x}_n$ to denote the set of variables appearing in the body of the query $q$. $C_1, ..., C_m $ are formulae of the form $x_i \oplus x_j$ where $\oplus$ is the comparison operator which could be anything from $<, \le, =, \not =, >, \ge$ and $x_i$ and $x_j$ are either variables in $\mathrm{Var}(q)$ or constants. 

We require that $\overline{y} \subseteq \mathrm{Var}(q)$. %
Without loss of generality, we assume that there are no self-joins in the query. In case of queries with self-joins, we can make logical copies of the relations to accommodate them~\cite{wang2022conjunctive}.
The body of a CQC $q$ comprises two parts, namely the relation identifiers $R_1, ..., R_n$ referred to as $\mathrm{ids}(q)$, and the conditions $C_1, ..., C_m$ denoted by $\mathrm{cnd}(q)$.

Similarly to Section~\ref{subsec:intro_to_databases}, the evaluation of $q$ on the database state $db$ (denoted by $\eval{q}{db}$) 
is defined by taking all tuples in the cartesian product of $\mathrm{ids}(q)$ in $db$ that satisfy $\mathrm{cnd}(q)$, and projecting to the column set $\overline{y}$.
\begin{example}
	Consider the database schema in Fig.~\ref{fig:relations_emp_man}. The following query returns a set of tuples containing the names of divisions whose managers have a salary of more than $50$:
	\begin{align*}
		\mathrm{ans}(d) \leftarrow \mathrm{emp}(n, r, s), \mathrm{mng}(d, m), n = m, s > 50
	\end{align*}
\end{example}

\subsection{Type-based Dependency Analysis}\label{sec:dep_analysis}
Our static dependency analysis builds on the generic type system of van Delft et al.~\cite{delft2015very} and extends it with support for disjunctive dependencies. We intuitively expect that a disjunctive dependency analysis must be path-sensitive, so as to distinguish between different executions and also keep track of the history of observations. Both of these requirements are often challenging for type-based analyses, which do not naturally align with the execution order. We will first illustrate these challenges with examples and then present our analysis.

\begin{figure}
\minipage[b]{0.49\columnwidth}
\begin{lstlisting}[label=prog:type_analysis_1]
if (y > 0) then
	x := w + z;
else
	x := x + 1;
out(x,u);
\end{lstlisting}
\endminipage\hfill
\minipage[b]{0.49\columnwidth}
\begin{lstlisting}[label=prog:type_analysis_2]
if (z == 0) then
	x $\leftarrow$ q1;
else
	x $\leftarrow$ q2;
out(x,u);
if (z != 0) then
	x $\leftarrow$ q1;
else
	x $\leftarrow$ q2;
out(x,u);
\end{lstlisting}
\endminipage\hfill
\end{figure}

Program~\ref{prog:type_analysis_1} illustrates the need for path sensitivity. The analysis should distinguish between the \emph{then} branch, where variable $x$ depends on the set $\{y,w,z\}$, and  the \emph{else} branch where $x$ depends on $\{y,x\}$. Our reference analysis \cite{delft2015very} would join these two sets at the end of the if statement, ultimately yielding the dependency set $\{x,y,w,z\}$. In our analysis, these sets are never joined, but instead combined to form a set of sets, namely, $\{\{y,w,z\}, \{y,x\}\}$, where the outer set represents a disjunctive dependency and the inner sets represent conjunctive dependency. 

Program~\ref{prog:type_analysis_2} illustrates the need to keep track of the observation history. It outputs $x$ at lines $5$ and $10$, and the dependency set of $x$ in both places is $\{\{q1,z\}, \{q2,z\}\}$. However, this program will always output both $q1$ and $q2$. Now, if a policy only allows user $u$ to see either query $q1$ or $q2$, the outputs at lines $5$ and $10$ will be incorrectly accepted. Hence, the analysis should account for all outputs to user $u$. 

Fig.~\ref{fig:type_analysis_rules} depicts the rules of our disjunctive dependency analysis. We use judgments of the form $\vdash c : \Gamma$, where $\Gamma$ is an environment mapping variables $\mathrm{Var}$ to set of sets of dependencies $\mathrm{Dep}$. The set of variables is  $\mathrm{Var} = PV \ \cup \ \users \ \cup \ \{pc\}$, where $PV$ are program variables, $\users$ are users, and $pc$ is the program context. The dependencies $\mathrm{Dep}$ are $\mathrm{Dep} = \mathrm{Var} \ \cup \ \queries$, where $\mathrm{Var}$ are variables and $\queries$ are queries that can be issued to a database. We use $u \in \users$ to indicate the dependencies of all outputs to user $u$.

We start by introducing the operators and auxiliary functions employed within the rules, and then proceed to explain the rules themselves. The operator $\otimes$ is used to join two (or more) sets of sets, defined as:
\begin{align*}
	\Gamma_1(x_1) \otimes ... \otimes \Gamma_n(x_n) = \{S_1 \cup ... \cup S_n \mid \ & S_i \in \Gamma_{i}(x_i)\\ &i=1, \dots,n \}
\end{align*}
For example, the join of $\Gamma_1(x) = \{\{x,y\}, \{z,y\}\}$ and $\Gamma_2(y) = \{\{w\}, \{x,z\}\}$ is:
\begin{align}
	\nonumber
	\Gamma_1(x) \otimes \Gamma_2(y) = \{\{x,y,w\}, \{x,y,z\}, \{z,y,w\}\}
\end{align}
Intuitively, the result of the join operator is a set of sets capturing the product of the original sets of sets under the set union operation.
We use this operator to calculate all the possible combinations of two environments.

$\Gamma_2;\Gamma_1$ represents the sequential composition of two environments. Intuitively, $\Gamma_2;\Gamma_1$ is the same as $\Gamma_2$ but updated with all of the dependencies that have been previously established in $\Gamma_1$. Formally:
\begin{align*}
	\nonumber
	\Gamma_2;\Gamma_1(x) = \bigcup\limits_{S_2 \in \Gamma_2(x)} \ \bigotimes\limits_{y \in S_2} \ \Gamma_1(y)
\end{align*}
For example, the sequential composition of the environments
\begin{align*}
	\Gamma_1 = [&x \mapsto \{\{x\},\{y\}\}, y \mapsto \{\{y\}\}, pc \mapsto \{\{y , pc\}\}] \\
	\Gamma_2 = [&x \mapsto \{\{pc,x\}\}, y \mapsto \{\{pc,y\}\}, pc \mapsto \{\{pc\}\}]
\end{align*}
evaluates to
\begin{align*}
	\Gamma_2;\Gamma_1 = [&x \mapsto \{\{x,y,pc\}, \{y,pc\}\}, y \mapsto \{\{pc,y\}\}, \\ &pc \mapsto \{\{y ,pc\}\}]
\end{align*}

Finally, the operator $\Cup$ calculates the union of two environments: $\Gamma_1 \Cup \Gamma_2 = \forall x \in \mathrm{Var}, \ \Gamma_1(x) \cup \Gamma_2(x)$.
This operator is used in conditionals to capture the disjunctive join of the two branches. For example, in line 5 in Program~\ref{prog:type_analysis_1}, $\Gamma_1(x) = \{\{y,w,z\}\}$ and $\Gamma_2(x) = \{\{y,x\}\}$, and the result of $(\Gamma_1 \Cup \Gamma_2)(x)$ would be $\{\{y,w,z\}, \{y,x\}\}$.

For loops, we rely on the fixed-point of $\Gamma$, denoted by $\Gamma^*$, which we define as:
\begin{align*}
	\Gamma^* = \bigcup\limits_{n > 0} \Gamma^{n}
\end{align*}
where $\Gamma^0 = \Gamma_{id}$ and $\Gamma^{n+1} = \Gamma^{n};\Gamma$.

In these rules, $\Gamma_{id}$ is the identity environment, defined as $\forall x \in \mathrm{Var}, \ \Gamma_{id}(x) = \{\{x\}\}$, and $\fv(e)$ denotes the free variables of expression $e$.

\begin{figure*}
	{
		\setstretch{1.2}
		\footnotesize
		\centering
		$
		\inferrule*[before=\textsc{T-Skip}]
		{
			\\
		}
		{
			\vdash \texttt{skip} : \Gamma_{id}
		}
		$
		\nextrule
		$
		\inferrule*[before=\textsc{T-Assign}]
		{
			\Gamma = \Gamma_{id}[x \mapsto \{\fv(e) \cup \{pc\}\}]
		}
		{
			\vdash x := e : \Gamma
		}
		$
		\nextrule
		$
		\inferrule*[before=\textsc{T-Output}]
		{
			\Gamma' = \Gamma_{id}[u \mapsto \{\fv(e) \cup \{pc , u\}\}]
		}
		{
			\vdash \texttt{out}(e,u) : \Gamma'
		}
		$
		
		\nextrule
		
		$
		\inferrule*[before=\textsc{T-QueryEval}]
		{
			\Gamma = \Gamma_{id}[x \mapsto \{\{q , pc\}\}]
		}
		{
			\vdash x \leftarrow q : \Gamma
		}
		$
		\nextrule
		$
		\inferrule*[before=\textsc{T-If}]
		{
			\vdash c_i : \Gamma_i \\
			\Gamma'_i = \Gamma_i;\Gamma_{id}[pc \mapsto \{\fv(e) \cup \{pc\}\}] \ i = 1,2 \\
			\Gamma' = (\Gamma'_1 \Cup \Gamma'_2)[pc \mapsto \{\{pc\}\}]
		}
		{
			\vdash \texttt{if} \ e \ \texttt{then} \ c_1 \ \texttt{else} \ c_2 : \Gamma'
		}
		$
		
		\nextrule
		
		$
		\inferrule*[before=\textsc{T-While}]
		{
			\vdash c : \Gamma_{c} \\
			\Gamma_{f} = (\Gamma_{c};\Gamma_{id}[pc \mapsto \{\fv(e) \cup \{pc\}\}])^* \\
			\Gamma' = \Gamma_{f}[pc \mapsto \{\{pc\}\}]
		}
		{
			\vdash \texttt{while} \ e \ \texttt{do} \ c : \Gamma'
		}
		$
		\nextrule
		$
		\inferrule*[before=\textsc{T-Seq}]
		{
			\vdash c_1 : \Gamma_1 \\
			\vdash c_2 : \Gamma_2 \\
			\Gamma' = \Gamma_2;\Gamma_1
		}
		{
			\vdash c_1;c_2 : \Gamma'
		}
		$
		
	}
	\caption{Type-based dependency analysis rules}
	\label{fig:type_analysis_rules}
\end{figure*}

\noindent\textsc{T-Assign} updates the dependency set of the assigned variable $x$ to the set of the free variables of expression $e$ and $pc$, otherwise it matches the identity environment. Rule \textsc{T-QueryEval} is similar to assignment, except that instead of $\fv(e)$, it adds query $q$ to the dependency set. 

\noindent\textsc{T-If} sequentially composes the dependency sets of each branch with the environment $\Gamma_{id}[pc \mapsto \{\fv(e) \cup \{pc\}\}]$, thus adding variables of the branch condition to the dependency set of each branch. Finally, these environments ($\Gamma_1$ and $\Gamma_2$) are joined disjunctively using the $\Cup$ operator.

\noindent\textsc{T-While} uses the fixed-point operator to calculate the dependency set of the loop. To do so, it first calculates the dependency set of the loop body, which is sequentially composed with $\Gamma_{id}[pc \mapsto \{\fv(e) \cup \{pc\}\}]$ to account for the dependencies to the loop condition. Finally, the fixed-point operator computes the dependency set of the while loop.

\noindent\textsc{T-Output} relies on the dependency set including $\fv(e)$, $\{pc\}$ and $\{u\}$, where $\fv(e)$ includes all the variables of the expression outputted to user $u$, $\{pc\}$ captures the implicit dependencies to the path conditions, and $\{u\}$ is the dependency set of user $u$ and captures the history of dependencies that user $u$ might have observed up to this point. Observe that by the definition of sequential composition, all the dependencies of the previous outputs will be added to $u$.

This analysis yields a final environment $\Gamma_{\mathrm{fin}}$. The result of the analysis is the value of this environment for the user identifier $u$, which includes both queries and program variables. 
Since program variables do not contain sensitive information, and we are primarily concerned with queries, we refine the result of $\Gamma_{\mathrm{fin}}(u)$ to only include queries. This refined outcome defines the ultimate result of our analysis, denoted as $\mathrm{QL}_u$:
\begin{align*}
	\mathrm{QL}_u \triangleq \bigcup_{S \in \Gamma_{\mathrm{fin}}(u)} \{ S \cap \queries  \}
\end{align*}

The soundness proof of our enforcement %
relies on the circumstance that, if the set of queries on which the $u$-outputs of $\prg$ depend when running on a database state $db$ are denoted by $Q_{\prg,u}(db)$, then this set is guaranteed to be found in the set $\mathrm{QL}_u$ produced by the dependency analysis.
We show how to define $Q_{\prg,u}(db)$ using a taint-tracking semantics presented in
{\iffull Appendix \ref{app:dep_analysis}. \fi} 
Formally, this gives rise to the following soundness condition for the dependency analysis.
\begin{globallemma}\label{lemma:type_analysis_soundness_proof}
For all $db\in \Omega_D$, $Q_{\prg,u}(db) \in \mathrm{QL}_u(\prg)$.
\end{globallemma}

\subsection{Query Abstraction}\label{sec:symbolic_tupes}
Even for CQCs, comparing the information revealed by sets of queries is hard in general.
To define a well-behaved and more tractable determinacy order on which to build our DQ, we introduce another
overapproximating abstraction, which we will use to soundly
\emph{label} queries and policies.

We define a \emph{symbolic tuple} as $\tuple{T,\phi, \pi}$, where $T = \{\tb_1, \tb_2 ...,\tb_n\} $ is a set of table identifiers, $\phi$ is a boolean combination of equality, inequality, and comparisons over the columns of the tables in $T$, and $\pi$ is a subset of the columns of the tables in $T$. In a symbolic tuple, $\pi$ denotes the query's projection on the columns of the tables in $T$, and $\phi$ defines the constraints over the rows.

\begin{example}
	The symbolic tuple of query $\mathrm{ans}(d) \leftarrow \mathrm{emp}(n,r,s), \mathrm{mng}(d,m), n = m, s > 50$ defined on the relations of Fig.~\ref{fig:relations_emp_man} would be $\tuple{ \{\mathrm{emp}, \mathrm{mng}\}, s > 50 \wedge n = m, \{d\} }$.
\end{example}

While calculating the exact set of symbolic tuples of a relational calculus query is intractable for many classes of queries, it is tractable for conjunctive queries with comparison (CQC). 
Given a conjunctive query $q = \mathrm{ans}(\overline{y}) \leftarrow R_1(\overline{x}_1), ..., R_n(\overline{x}_n), C_1, ..., C_m $, the function $\mathrm{sts}$ computes a symbolic tuple from $q$ as follows:
\begin{align*}
	\mathrm{sts}(q) = \tuple{\mathrm{ids}(q'), \big( \bigwedge_{C \in \mathrm{cnd}(q')} C \ \big) , \overline{y}}
\end{align*}
where $\mathrm{ids}(q')$ and $\mathrm{cnd}(q')$ defined in Section~\ref{sec:CQC} return the relation identifiers and conditionals of $q'$, respectively. Here, $q'$ is the query obtained by recursively replacing views with their definitions. We lift this definition to sets of queries $Q$, and define $\mathrm{sts}(Q)$ as $\{ \bigcup_{q \in Q} \mathrm{sts}(q) \}$. 

Using $\mathrm{sts}$, we define the function $\sigma_{\st}$ for a set of sets of queries $\Q$ as follows:
\begin{align*}
	\sigma_{\st}(\Q) = \{ \mathrm{sts}(Q) \mid Q \in \Q \}
\end{align*}

\tightpar{Policy Analysis}
The function $\sigma_{st}$ can also be used to map a disjunctive security policy to a set of labels. 
However, in order to ensure soundness and avoid approximation, we place some constraints on policies. (1) To make computing the set of symbolic tuples tractable we only support policies with views in the CQC form. (2) We require that the symbolic tuples of views be \mbox{\emph{well-formed}}, which we define as:
\begin{definition}\label{def:st_well_formedness}
	The symbolic tuple $\tuple{T,\phi, \pi}$ is said to be well-formed if it satisfies $\dep(\phi) \subseteq \pi$.
\end{definition}
\noindent where $\phi = C_1 \wedge ... \wedge C_n$ and $\dep(\phi) = \bigcup_{i \in \{1,...,n\}} \fv(C_i)$ returns the column dependency set of $\phi$. 

Well-formedness ensures that the symbolic tuples are precise, at the expense of limiting a view to only applying constrains on the columns which it projects on. 

Furthermore, we treat the table identifiers used in policies as special views that return the whole table. For instance, a policy which allows access to table $\mathrm{emp}$ can be rewritten as view $\mathrm{ans}(n,r,s) \leftarrow \mathrm{emp}(n,r,s)$. 

As discussed in Section~\ref{sec:system_model}, the disjunctive security policy of user $u$ (written as $P_u$) is a set of conjunctions $\con$, interpreted as a disjunction of conjunctions of table and view identifiers. For a policy $P_u$ that adheres to the constraints mentioned earlier, $\sigma_{\st}$ is defined as follows:
\begin{align*}
	\sigma_{\st}(P_u) = \{ \mathrm{sts}(\con) \mid \con \in P_u \}
\end{align*}

\tightpar{Labels} 
In our model, a security label $\ell$ is defined as a set of symbolic tuples, and we define the ordering relation of two labels, written as $\ell_1 \stsubset \ell_2$, as follows: 
\begin{definition}\label{def:symbolic_tuples_ordering}
	$\ell_1 \stsubset \ell_2$ iff for all symbolic tuples $\tuple{T,\phi, \pi} \in \ell_1$, there are well-formed symbolic tuples $\tuple{T_1,\phi_1, \pi_1},...,\tuple{T_n,\phi_n, \pi_n}$ in $\ell_2$ such that $T \subseteq (T_1 \cup ... \cup T_n)$, $T_1,..., T_n$ are disjoint, $\phi \models (\phi_1 \wedge ... \wedge \phi_n)$, and $dep(\phi) \cup \pi \subseteq (\pi_1 \cup ... \cup \pi_n)$.
\end{definition}
To ensure soundness, we assume that all of the symbolic tuples in the right hand side of $\stsubset$ are well-formed.  
This definition relies on entailment to check the ordering of $\phi$, and write $\phi_1 \models \phi_2$ which means that any assignment that satisfies $\phi_1$ also satisfies $\phi_2$. 

\begin{example} Consider symbolic tuples $\ell_1 = \{\tuple{ \{\mathrm{emp}\}, s = 10, \{r\} }\} $ and $\ell_2 = \{ \tuple{ \{\mathrm{emp}, \mathrm{mng}\}, s > 5, \{r,s,m\} } \} $. We have $\ell_1 \stsubset \ell_2$ since $\{\mathrm{emp}\} \subseteq \{\mathrm{emp}, \mathrm{mng}\}$, $\{r\} \subseteq \{r,s,m\}$, $s = 10 \models s > 5$ and $\{s\}\cup\{r\} \subseteq \{r,s,m\}$. 
\end{example}

\subsection{Enforcement} \label{sec:enforcement_security_check}
The dependency analysis of Section~\ref{sec:dep_analysis} extracts the dependencies of program $\prg$'s outputs to user $u$ and produces $\mathrm{QL}_u$. Applying $\sigma_{\st}$ to $\mathrm{QL}_u$ yields a set of labels, each bounding the information revealed in some path, the $u$-knowledge of $\prg$ (denoted by $k(\prg)_u$). We interpret this as a disjunction, as any execution follows along one particular path.

Similarly, applying $\sigma_{\st}$ to the disjunctive security policy of user $u$ (\ie $P_u$) results in a set of labels. Each label faithfully captures one conjunction, and so the policy is also represented as a set of labels $ak(P_u)$, interpreted disjunctively.

By Lemma~\ref{lemma:DQ_imply_QoI}, to verify that the security condition is satisfied, it is sufficient to establish that $\mathrm{QL}_u \sqsubseteq P_u$ in the
DQ. However, checking $\sqsubseteq$ in the DQ is not generally tractable.
For the security check, we therefore instead perform a twofold approximation: we check ordering for the conjunctive inner sets using the approximate ordering $\stsubset$, and approximate the mix-based ordering on the disjunctive outer sets in a way that loses little relative to our analysis:
\begin{definition} \label{def:secuirty_check}
	We say that $k(\prg)_u \sqsubseteq_* ak(P_u)$ iff
	\begin{align*}
		\forall \ell_{k} \in k(\prg)_u, \ \exists \ell_{ak} \in ak(P_u) . \ \ \ell_{k} \stsubset \ell_{ak}
	\end{align*}
\end{definition}
\noindent where $\ell_{ak} $ and $\ell_{k}$ are labels, and $\stsubset$ is the symbolic tuple ordering of Def.~\ref{def:symbolic_tuples_ordering}. To ensure faithful labeling of policies, we assume all of the symbolic tuples in $\ell_{ak}$ are well-formed as defined in Def.~\ref{def:st_well_formedness}. 
We can then formalize the relationship between $\sqsubseteq_*$ and $\sqsubseteq$ as follows.
\begin{restatable}{globallemma}{unfoldingAndQuantaleProof}\label{lemma:unfolding_and_quantale}
	If $\sigma_{\st}(\{Q_1, ..., Q_n\}) \sqsubseteq_* \sigma_{\st}(\{P_1, ..., P_m\})$, then in the DQ,
$		(Q_1 \vee ... \vee Q_n) \sqsubseteq (P_1 \vee ... \vee P_m)$.
\end{restatable}
{\iffull
We refer the readers to Appendix~\ref{sec:app:proof:symbolic_tuple_DQ_ordering} for the proof of this Lemma.
\fi}

\subsection{Soundness Proof}\label{sec:enforcement_proof}
Fig.~\ref{fig:proof_steps} outlines the overall architecture of our enforcement mechanism and the correctness assertion that we make of it.

\begin{figure}
	\centering
	\begin{tikzpicture}
		\begin{scope}[xscale=2.2, yscale=1.4] %
			\node (lp1) at (0,-0) {$ak(P_u)$};
			\node (lk1) at (0,-1) {$k(\prg)_u$};
			\node (lp2) at (1,-0) {$P_u$};
			\node (lk2) at (1,-1) {$\text{QL}_u$};
			\node (lp3) at (2,-0) {$\llbracket P_u \rrbracket$};
			\node (lk3) at (2,-1) {$\llbracket \text{QL}_u \rrbracket$};
			\node (prog) at (1,-2) {$\prg$};
			\node (progi) at (2,-2) {$\llbracket \prg \rrbracket_u$};

            \node (impl1) at (0.5,-0.5) {$\Rightarrow$};            
	        \node (impl2) at (1.5,-0.5) {$\Rightarrow$};            
	        \node (impl3) at (1.5,-1.5) {$\Rightarrow$};            
	
			\begin{scope}[every node/.style={rotate=90}]
				\node (cmp) at ($(progi)!0.5!(lk3)$) {$\sqsubseteq$};
				\node (cmp) at ($(lk3)!0.5!(lp3)$) {$\sqsubseteq$};
				\node (cmp) at ($(lk2)!0.5!(lp2)$) {$\sqsubseteq$};
				\node (cmp) at ($(lk1)!0.5!(lp1)$) {$\sqsubseteq_*$};
			\end{scope}
			
			\draw[|->] (prog)--(progi);
			\draw[|->] (lk2)--(lk3);
			\draw[|->] (lp2)--(lp3);
			\draw[|->] (prog)--node[left] {\footnotesize D.A.}(lk2);
			\draw[|->] (lk2)--node[below] {\footnotesize $\sigma_{\st}$}(lk1);
			\draw[|->] (lp2)--node[above] {\footnotesize $\sigma_{\st}$}(lp1);
			\draw[<-|] (lk1)--(lk2);
			
			\node[draw, thick, dotted, rounded corners, inner xsep=0.5em, inner ysep=0.6em, fit=(lp3) (progi)] (box) {};
			\node[fill=white] at (box.north) {QoI};
			
			\node[draw, thick, dotted, rounded corners, inner xsep=0.5em, yshift=0em, inner ysep=0.6em, fit=(lp2) (lk2)] (box) {};
			\node[fill=white] at (box.north) {DQ};
			
		\end{scope}
	\end{tikzpicture}
	
	\caption{Overall architecture of our proof}
	\label{fig:proof_steps}
\end{figure}
The rightmost column of Fig.~\ref{fig:proof_steps} represents a chain of information order relations in the QoI, which we establish for each enforcement step. Following the chain from bottom to top, we obtain the security condition of Def.~\ref{def:security_condition}.
At the same time, the ``left boundary'' of the figure, comprising the D.A., $\sigma_{\st}$ abstractions
and $\sqsubseteq_*$ check, represents the computations that are actually performed to check a program.

\begin{restatable}{globaltheorem}{enforcementSoundnessProofHighLevel}\label{thrm:enforcement_soundness_proof_highlevel}
If a program $\prg$ satisfies Def.~\ref{def:secuirty_check}, then it is secure in the sense of Def.~\ref{def:security_condition}.
\end{restatable}
{\iffull
\begin{proof} The statement follows from establishing the implications in the diagram of Fig.~\ref{fig:proof_steps}. The top left cell is Lemma~\ref{lemma:unfolding_and_quantale}; the top right cell is Lemma \ref{lemma:DQ_imply_QoI}; and the bottom cell (dependency analysis) is Appendix \ref{app:dep_analysis}.
\end{proof}
\fi}

\section{Implementation and Evaluation}\label{sec:implementation}
In this section, we describe our prototype \toolname~\cite{divertTool}, which implements the type-based dependency analysis of Section~\ref{sec:dep_analysis} and query abstraction of Section~\ref{sec:symbolic_tupes} to verify the security of database-backed programs. We then  evaluate \toolname's effectiveness using functional tests and an assortment of real-world-inspired use cases.

\subsection{Implementation}
To evaluate the feasibility and security of our approach in practice, we implemented the type-based dependency analysis of Section~\ref{sec:dep_analysis}. 
For the sake of practicality, instead of CQC, \toolname uses the \texttt{SELECT}-\texttt{FROM}-\texttt{WHERE} portion of SQL, which is analogous to CQC as described in Section~\ref{sec:CQC}. Following the query analysis of Section~\ref{sec:symbolic_tupes}, these SQL queries are then converted into symbolic tuples. For the security check, the symbolic tuples with the result of the program analysis must be compared to those representing the policy; to perform this comparison following Def.~\ref{def:symbolic_tuples_ordering}, we use the Z3 SMT solver~\cite{demoura2008z3}.
Our implementation operates on programs in the language presented in Section~\ref{sec:language}, with the addition of two macros \texttt{@Table@} and \texttt{@Policy@} for defining the tables' schema and the security policy.

\subsection{Test suite}
To validate our implementation, we use a functional test suite consisting of \benchmarkNumber programs, designed to capture a broad variety of examples of disjunctive dependencies.
This suite includes programs with row- and column-level policies of varying granularity levels, and those necessitating the use of SMT solvers for verification. Furthermore, the tests verify the behaviour of the dependency analysis by incorporating complex conditionals, loops, and implicit and explicit outputs. 
The tests can be found in the implementation repository~\cite{divertTool}.

\subsection{Use cases}
We evaluate \toolname on four use cases inspired by real-world problems in which disjunctive policies naturally arise. %
The purpose of this evaluation is to validate the security analysis of \toolname on realistic scenarios involving disjunctive policies, and ensure that its behaviour is consistent with the definitions of Section~\ref{sec:security_model}. Rather than analysing complete applications for each example, we therefore focus on smaller kernels that capture the core security-critical behaviour of the respective problem.

\tightpar{Privacy-preserving location service}
Multilateration is a technique to determine the location of a user by measuring their distance to known reference points~\cite{murphy1995determination}. Two distances are sufficient to narrow a user's location down to one of two points on a map, and three identify the location unambiguously.
Consider a location service provider which tracks, for some number of users, not only their precise location but also their distances to certain points of interest (PoI) such as restaurants or shops. An advertiser wants to query this service to provide location-based ads. For example, if the user is close to a shop $A$, and $A$ has a sale going on, the user may be enticed by this information.

Privacy and business considerations make it desirable to not reveal the precise location of the user to the advertisement company accessing the database, while still allowing for some location-based services in this vein.
If the advertiser were to learn the distance of a single user to two or more PoIs at a specific time, the user's location could be inferred. However, we may still want to release the user's distance to any one PoI which they are currently closest to. 
This can be interpreted as a disjunctive policy, in which the information revealed for each user is bounded by the disjunction of that user's distances to some \emph{single} PoI.

The database schema consists of a single table \texttt{Distance(id, poi, dis, loc)}, which stores the ID of each user, the name of the PoI, their distance, and the user's precise location.
We implement a small example with two PoIs $\{ \text{`restaurant'}, \text{`mall'} \}$ and two users $\{ 1, 2\}$.
Let the view $\vw_{i,j}$ for each user $i$ and PoI $j$ be defined as the query \texttt{SELECT id, poi FROM Distance WHERE id = $i$ AND poi = $j$}. The disjunctive policy then covers every combination of user and PoI as a possibility: $\{ \{\vw_{1,\text{`restaurant'}}, \vw_{2,\text{`restaurant'}}\}, \{\vw_{1,\text{`restaurant'}}, \vw_{2,\text{`mall'}}\},$ $\{\vw_{1,\text{`mall'}}, \vw_{2,\text{`restaurant'}}\}, \{\vw_{1,\text{`mall'}} ,\vw_{2,\text{`mall'}}\}\}$.

We test two programs against this policy. In one, the advertiser uses internal parameters identifying a target user and interest, and issues a single query requesting that user's distance from the relevant point of interest. In the other, the advertiser still targets a particular user, but queries all of that user's distances. As expected, \toolname accepts the former program, but rejects the latter.

\tightpar{Privacy-preserving data publishing} 
Expanding upon the motivating example in the introduction, we consider the case of programs querying a database with personally identifiable information (\ie quasi-identifiers).
As discussed before, revealing too many quasi-identifiers may make it possible to identify an individual. 
We consider the example of a medical database~\cite{sweeney2002kanonymity} with a table \texttt{Patients(zip, gen, dis)} storing the ZIP code of residence, gender and disease of patients.
An agent querying the database should not learn more than two of these at a time. For simplicity's sake, we only consider queries that retrieve the same data from each patient.
Defining
$\vw_1=$ \texttt{SELECT dis, gen FROM Patients},
$\vw_2=$ \texttt{SELECT zip, gen FROM Patients}, and
$\vw_3=$ \texttt{SELECT zip, dis FROM Patients},
the disjunctive policy can then be written as $\{ \{\vw_1\}, \{\vw_2\}, \{\vw_3\} \}$.

Once again, we validate two programs against this policy. Branching
on an internal parameter, the client will issue one query to select data for either male or female patients.
In the first program, all queries take the form of \texttt{SELECT dis FROM Patients WHERE gen = $\squote{F}$},
whereas in the second one, one of the queries additionally filters on the ZIP code:
\texttt{SELECT dis FROM Patients WHERE gen = $\squote{F}$ AND zip = $10001$}. Again, only the latter
program is rejected by \toolname. This reveals a potential subtlety, as data dependency and hence release of information may arise not only from what columns are selected, but also from conditions restricting the set of rows.

\tightpar{Secret sharing}
We implement a $(t,n)$ secret sharing schema that splits a secret value $s$ into $n$ shares $s_1, s_2, ..., s_n$. These shares are then distributed among $n$ parties $p_1, p_2, ..., p_n$, each  receiving a unique share. A secure secret sharing schema requires that the secret $s$ can only be reconstructed if $t$ or more participants combine their shares. If the number of combined shares is less than $t$, no information about the secret should be revealed.
This requirement naturally translates to a disjunctive policy  $s_1 \vee s_2 \vee ... \vee s_n$, stipulating that participants can each only learn \emph{one} share. 

We assume that the shares $s_1, s_2, ..., s_n$ are created by a secure secret sharing schema and are then stored in a database. The database schema consists of the table \texttt{Shares(shareID, shareVal)} which stores the ID of each share and their corresponding value. 

The policy only allows a user to read one of the shares (\ie only one row of the table). We define the view $\vw_i$ for each share as \texttt{SELECT shareVal, shareID FROM Shares WHERE shareID = $i$} where $i=1,...,n$. The corresponding disjunctive policy is going to look like $\{ \{\vw_1\}, \{\vw_2\}, ..., \{\vw_n\} \}$.

We implement a program that executes a subroutine for each user, issuing a database query to retrieve the user's share. For example the query for a user to retrieve the share number $5$ is \texttt{SELECT shareVal FROM Shares WHERE shareID = $5$} and it is correctly accepted by \toolname. If the same user issues another query to retrieve share number $6$, it violates the policy and hence the program is rejected. This scenario shows that \toolname is able  to correctly enforce row-level policies precisely.

\tightpar{Online shop}
This use case models an online shop and a user with a gift card can only use it to ``buy" items that match the value of the gift card.
Here we consider a scenario with an online shop that only provides digital items and they are stored in a database. The database schema consists of the items table \texttt{Items(id, name, data)} which stores the ID and name of each digital item. We define a view $\vw_n$ for each item as \texttt{SELECT data, name FROM Items WHERE name = $n$} where $n$ is the item's name. 

Assume a database that has the  items \emph{Movie}, \emph{CinemaTicket}, \emph{Audiobook}, \emph{Ebook}, and \emph{GymMem}. A policy  should only allow the user to access a certain amount of items whose value adds up to value of gift card. For instance a disjunctive policy may look like: $\{ \{\vw_{\mathrm{Movie}}, \vw_{\mathrm{CinemaTicket}}\}, \{\vw_{\mathrm{Audiobook}}, \vw_{\mathrm{Ebook}}\}, \{\vw_{\mathrm{GymMem}}\},$ $\{\vw_{\mathrm{CinemaTicket}}, \vw_{\mathrm{Ebook}}\}  \}$.

We model a user program that issues queries to select items, e.g., \texttt{SELECT data FROM Items WHERE name = $\squote{Movie}$}.

\toolname accepts this query because view $\vw_{\mathrm{Movie}}$ allows the user to access \emph{Movie}. We create two different scenarios; in one the user issues another query asking for \emph{Audiobook}, which \toolname rejects. In the second scenario, the user asks for \emph{CinemaTicket} which is allowed by the policy, and hence \toolname accepts it.

\section{Related Work}\label{sec:related_work}
This section puts our contributions in the context of related works in the areas of information flow security and database security, discussing security models of dependencies and tractable enforcement mechanisms. To our knowledge, we are the first to explore enforcement mechanisms for disjunctive policies, as well as to reconcile semantic models of (disjunctive) dependencies across the areas of information flow control and database access control.

\tightpar{Security models} Semantic models of dependencies have a long history since the introduction of the Lattice of Information (LoI) by Landauer and Redmond \cite{landauer1993lattice}. These models define a lattice structure to represent information as equivalence relations ordered by refinement and  serve as  cornerstone to justify soundness of various dependency analysis at the heart of enforcement mechanisms for security. 
For example, the universal lattice by Hunt and Sands \cite{DBLP:conf/popl/HuntS06} models dependencies between program variables such that the lattice elements are sets of variables ordered by set containment, and uses it to justify soundness against baseline security conditions, e.g., noninterference \cite{GoguenM82}. 

Within the database community, Bender et al. \cite{bender2013fine,bender2014explainable} define the notion of Disclosure Lattice to represent the information disclosed by sets of database queries. 
Disclosure Lattice has been further developed by Guarnieri et al. \cite{guarnieri2019information} to enforce conjunctive information-flow policies for database-backed programs. We point out that not all disclosure orders are suitable to represent information disclosure in the context of information flow control: By studying its relation to LoI, we show that query determinacy and the stronger notion of equivalent query rewriting \cite{nash2010views} provide sound abstraction, while query containment does not. 

Our work builds on recent work by Hunt and Sands \cite{hunt2021quantale}, which provides a semantic model for disjunctive dependencies, under the notion of the Quantale of Information. %
We study quantale structures in the context of databases, providing support for disjunctive policies in database-backed programs. 
While these policies are rooted in the area of access control, cf. ethical wall policies \cite{nash89}, the work of  Hunt and Sands \cite{hunt2021quantale} is the first to provide an extensional characterization as information-flow policies. Drawing on our new notion of Determinacy Quantale, we develop a security condition to capture the security of database-backed programs in presence of disjunctive database policies.

\tightpar{Enforcement mechanisms} The problem of enforcing disjunctive policies for programs and/or databases is completely unexplored. We study how a standard type-based program analysis \cite{delft2015very}, equipped the notion of path sensitivity, can be adapted to statically capture disjunctive program dependencies.  

At the core of our analysis is a new abstraction of database queries which enables flexible enforcement of disjunctive policies by means of SMT solvers, as witnessed by our use cases. An immediate benefit of our Determinacy Quantale is that we can prove soundness of the enforcement with respect to a solid semantic baseline for disjunctive dependencies. 

There exists a wide array of works enforcing conjunctive policies for database-backed programs.  Guarnieri et al. \cite{guarnieri2019information} propose dynamic monitoring to enforce database policies. Their abstractions are limited to boolean queries and rely on the Disclosure Lattice of  Bender et al. \cite{bender2013fine,bender2014explainable}, which may cause soundness issues when assuming query containment as the underlying lattice order. 
 
Language-integrated queries are supported by a range of works  such as SIF~\cite{chong2007sif} and \textsc{JsLinq}~\cite{balliu2016jslinq}, \textsc{SeLinks}~\cite{corcoran2009cross}, UrFlow~\cite{DBLP:conf/osdi/Chlipala10}, DAISY~\cite{guarnieri2019information}, Jacqueline~\cite{yang2016precise}, and LWeb~\cite{DBLP:journals/pacmpl/ParkerVH19} for  row- and column-level conjunctive policies. These works apply PL-based enforcement techniques such as type systems, dependent types, refinement types, and symbolic execution to database-backed programs~\cite{DBLP:conf/sp/SwamyCH08,DBLP:conf/popl/LourencoC15,DBLP:journals/pacmpl/ParkerVH19,guarnieri2019information}, but lack support for expressing and enforcing disjunctive policies. 

Li and Zhang \cite{DBLP:conf/csfw/LiZ17} explore path-sensitive program analysis to improve precision of information flow analysis, yet they do not consider disjunctive policies.
QAPLA \cite{DBLP:conf/uss/MehtaEH0D17} is a database access control middleware supporting complex security policies, such as linking and aggregation
policies, with focus only on access control.

\section{Conclusions}\label{sec:conclusion}

We presented a case for the significance of disjunctive dependency analysis to the security of database-backed programs. After reviewing recent theoretical developments in representing disjunctive information, we introduced two structures, the Determinacy Lattice and the Determinacy Quantale, as database-oriented counterparts to theoretical structures representing simple and disjunctive knowledge respectively. 

Using these structures, we formulated a security condition which expresses that a database-backed program satisfies a given disjunctive policy. 
In order to enforce this security condition, we developed a type-based static analysis to compute a bound on the disjunctive dependencies of database-backed programs in a model language. By a series of approximations, this bound itself can be tractably compared to the representation of a static policy.

These steps constitute an enforcement mechanism for disjunctive policies, which we proved sound with respect to our security condition.
To showcase this enforcement mechanism, we implemented it in our prototype tool, \toolname. In order to validate this prototype and the overall framework, we verified the tool on a set of functional tests covering a variety of language features and disjunctive information patterns, as well as several use cases representing real-world scenarios in which we want to enforce disjunctive policies.

\section{Acknowledgements}

We are grateful to David Sands and Roberto Guanciale for fruitful discussions, and would also like to thank the anonymous reviewers for their insightful comments and feedback.

This work was partially supported by the Wallenberg AI, Autonomous Systems and Software Program (WASP) funded by the Knut and Alice Wallenberg Foundation, the Swedish Research Council (VR), and the Swedish Foundation for Strategic
Research (SSF).

\bibliographystyle{IEEEtran}
\bibliography{bibliography}

\iffull
\appendices
\section{Interpretations of Query Determinacy}\label{sec:app:proof:determinacy_def_equiv}
We prove the following technical lemma to show that
the two intuitive interpretations of the definition of query determinacy are equivalent.

\begin{globallemma}\label{lemma:app:determinacy_def_equiv}
	If $A$ is recursively enumerable and $f:A\rightarrow B$ and $g:A \rightarrow C$ are computable, then the following are equivalent:
	\begin{itemize}
		\item[(i)] For all $a,a'\in A$, if $f(a)=f(a')$, then $g(a)=g(a')$.
		\item[(ii)] There exists a computable $h:B\rightarrow C$ such that for all $a\in A$, $g(a)=h(f(a))$.
	\end{itemize}
\end{globallemma}
\begin{proof}
	\noindent (ii)$\Rightarrow$(i): Suppose $b=f(a)=f(a')$, and $h$ is as in (ii). Then
	$g(a)=h(f(a))=h(b)$, and $g(a')=h(f(a'))=h(b)$.
	
	\noindent (i)$\Rightarrow$(ii): Let $\hat f:B\rightharpoonup A$ be the
	partial function that enumerates $A$ and for a given $b\in B$ returns the first $a\in A$ it finds
	such that $f(a)=b$. This is computable, per the algorithmic description provided.
	This does not necessarily satisfy $\hat f(f(a))=a$, but we do have
	$f (\hat f (f(a)))=f(a)$ by definition (since the enumeration of $A$ will either encounter $a$ or
	another $a'$ such that $f(a')=f(a)$ eventually). Hence
	$g (\hat f (f(a)))=g(a)$ by (i). So defining $h$ by $h(b) = g(\hat f(b))$,
	we find that $h(f(a))=g(a)$ as required.
\end{proof}

Instantiating Lemma~\ref{lemma:app:determinacy_def_equiv} with $A$ as the set of possible databases,
$f$ as the function $r_Q(db) = \{ \eval{q}{db} \mid q\in Q \}$ that computes the results of the queries in $Q$ on $db$, and $g$ as the same for $Q'$, we find that $Q$ determining $Q'$ indeed means that the (results of) queries in $Q$ are always sufficient to determine (compute) the result of the queries in $Q'$.

\section{Relation Between DL and LoI}\label{sec:app:proof:determinacy_loi_lattice}
We first prove some auxiliary lemmas, and then proceed to prove Lemma~\ref{lemma:DL_imply_LoI}.

\begin{globallemma}\label{lemma:app:ordering_loi_determinacy_lattice}
	For sets of queries $Q_1 ,Q_2 \in DL(\queries)$, the ordering $\cl{Q_1} \sqsubseteq \cl{Q_2}$ on the DL implies $\eq{Q_1} \sqsubseteq \eq{Q_2}$ on the LoI defined on $\{\eq{Q} \mid Q \in DL(\queries)\}$:
	\begin{align*}
		\cl{Q_1} \sqsubseteq \cl{Q_2} \rightarrow \eq{Q_1} \sqsubseteq \eq{Q_2}  
	\end{align*}
\end{globallemma}
\begin{proof}
	The definition of the ordering relation of the LoI (Section~\ref{sec:preliminaries}) and $\eq{Q_1} \sqsubseteq \eq{Q_2}$ would give us:
	\begin{align*}
		\eq{Q_1} \sqsubseteq ~&\eq{Q_2} \rightarrow \\
		& \forall db,db' \in \Omega_D \ \ (db \ \eq{Q_1} \ db' \Rightarrow db \ \eq{Q_2} \ db') \tag{1}
	\end{align*}
	
	By the definition of equivalence relations for query sets ($\eq{Q}$), for all $db,db' \in \Omega_D $ we have:
	\begin{align*}
		&(db \ \eq{Q_1} \ db' \Rightarrow db \ \eq{Q_2} \ db') \rightarrow \\
		&\Big( (\eval{q_2}{db} = \eval{q_2}{db'} \forall q_2 \in Q_2) \Rightarrow (\eval{q_1}{db} = \eval{q_1}{db'} \forall q_1 \in Q_1) \Big)\tag{2}
	\end{align*}
	(1) and (2) would give us:
	\begin{align*}
		&\eq{Q_1} \sqsubseteq \eq{Q_2} \rightarrow \forall db,db' \in \Omega_D \\ 
		&\Big( (\eval{q_2}{db} = \eval{q_2}{db'} \forall q_2 \in Q_2) \Rightarrow (\eval{q_1}{db} = \eval{q_1}{db'} \forall q_1 \in Q_1) \Big)\tag{3}
	\end{align*}
	
	On the other hand, by the definition of the Determinacy Lattice~\ref{def:determinacy_lattice}, we have $\cl{Q_1} \sqsubseteq \cl{Q_2} \leftrightarrow Q_1 \preceq Q_2$. From the definition of determinacy ordering, $Q_1 \preceq Q_2$ means $Q_2 \twoheadrightarrow Q_1$. 
	By the definition of query determinacy (Def.~\ref{def:query_determinacy}) we know that $Q_2 \twoheadrightarrow Q_1$ if:
	\begin{align*}
		&\forall db,db' \in \Omega_D \\
		&\Big( (\eval{q_2}{db} = \eval{q_2}{db'} \forall q_2 \in Q_2) \Rightarrow (\eval{q_1}{db} = \eval{q_1}{db'} \forall q_1 \in Q_1) \Big) \tag{4}
	\end{align*}
	
	It is evident from (3) and (4) that $\cl{Q_1} \sqsubseteq \cl{Q_2} \rightarrow \eq{Q_1} \sqsubseteq \eq{Q_2}$ holds.
\end{proof}

Relying on Def.~\ref{def:set_query_to_set_eq} to establish the set of equivalence relations derived from a set of sets of queries, we propose following lemma:

\begin{globallemma}\label{lemma:app:join_loi_determinacy_lattice}
	For any set of sets of queries $\Q \subseteq DL(\queries)$, the join of $\Q$ on the DL implies the join of $\EQ{\Q}$ on the LoI defined on $\{\eq{Q} \mid Q \in DL(\queries)\}$:
	\begin{align*}
		\bigsqcup \Q \rightarrow \bigsqcup \EQ{\Q}
	\end{align*}
\end{globallemma}
\begin{proof}
	Assume there is a set of queries $R \in DL(\queries)$ such that $R = \bigsqcup \Q$.
	
	By the definition of the Determinacy Lattice~\ref{sec:determinacy_lattice}, we have $\bigsqcup \Q = \cl{\bigcup \Q}$ which would give us $R = \cl{\bigcup \Q}$.
	By the definitions of $\clEmpty$ and query determinacy(Def.~\ref{def:query_determinacy}), it is straightforward to see $(\bigcup \Q) \twoheadrightarrow \cl{\bigcup \Q}$ and $\cl{\bigcup \Q} \twoheadrightarrow (\bigcup \Q)$.
	Replacing $\cl{\bigcup \Q}$ with $R$, by the definition of query determinacy (Def.~\ref{def:query_determinacy}) we have $R \twoheadrightarrow (\bigcup \Q)$:
	\begin{align*}
		&\forall db,db' \in \Omega_D \\
		&\Big( \forall r \in R. \ \eval{r}{db} = \eval{r}{db'} \rightarrow \forall p \in \bigcup \Q. \ \eval{p}{db} = \eval{p}{db'} \Big) \tag{1}
	\end{align*}
	\noindent and $(\bigcup \Q) \twoheadrightarrow R$:
	\begin{align*}
		&\forall db,db' \in \Omega_D \\
		&\Big(\forall p \in \bigcup \Q. \ \eval{p}{db} = \eval{p}{db'} \rightarrow \forall r \in R. \ \eval{r}{db} = \eval{r}{db'} \Big) \tag{2}
	\end{align*}
	(1) and (2) would give us:
	\begin{align*}
		&\forall db,db' \in \Omega_D \\
		&(\forall r \in R. \ \eval{r}{db} = \eval{r}{db'} \leftrightarrow \forall p \in \bigcup \Q. \ \eval{p}{db} = \eval{p}{db'}) \tag{3}
	\end{align*}

	Assume $\bigsqcup \EQ{\Q}$ is an equivalence relation $R'$. By the definition of the join of the LoI (Section~\ref{sec:preliminaries}):
	\begin{align*}
		\bigsqcup \EQ{\Q} = \forall db,db' \in \Omega_D \ (db \ \eq{R'} \ db' \leftrightarrow \forall Q \in \Q. \ db \ \eq{Q} \ db')
	\end{align*}
	and by the definition of equivalence relations for query sets, for all $db,db' \in \Omega_D $ we have:
	\begin{align*}
		&\bigsqcup \EQ{\Q} = \\
		&(\forall r \in R'. \ \eval{r}{db} = \eval{r}{db'} \leftrightarrow \forall Q \in \Q. \ \forall q \in Q. \ \eval{q}{db} = \eval{q}{db'}) = \\
		&(\forall r \in R'. \ \eval{r}{db} = \eval{r}{db'} \leftrightarrow \forall p \in \bigcup \Q. \ \eval{p}{db} = \eval{p}{db'}) \tag{4}
	\end{align*}
	
	(3) and (4) would allow us to conclude $R = R'$, hence $\bigsqcup \Q \rightarrow \bigsqcup \EQ{\Q}$.
\end{proof}

\determinacyLoILattice*
\begin{proof}
	To prove this homomorphism, we need to show that the Determinacy Lattice's ordering and join, as well as the top and bottom elements imply their LoI counterparts. Lemmas~\ref{lemma:app:ordering_loi_determinacy_lattice} and \ref{lemma:app:join_loi_determinacy_lattice} provide the proofs of ordering and join. 
	The proof for top and bottom elements: 
	\begin{itemize}
		\item $\cl{\queries} \rightarrow \eq{(\cl{\queries})}$
		\item $\cl{\varnothing} \rightarrow \eq{(\cl{\varnothing})}$
	\end{itemize}
	follows trivially from the definition of $\clEmpty$ and $\eq{}$.
\end{proof}

\section{Determinacy Quantale Axioms}\label{sec:app:proof:determinacyQuantale}
We follow the approach of \cite{hunt2021quantale} to prove that our definition of the Determinacy Quantale is indeed a quantale. We begin by defining what is a quantale.
\begin{definition}\label{def:quantale_axioms}
	A quantale is a structure $\tuple{\mathcal{L}, \sqsubseteq, \vee, \otimes, 1}$ such that:
	\begin{enumerate}
		\item $\tuple{\mathcal{L}, \sqsubseteq, \vee}$ is a complete join-semilattice
		\item $\tuple{\mathcal{L}, \otimes, 1}$ is monoid, that is $\otimes$ is associative and $\forall x \in \mathcal{L}, x \otimes 1 = x = 1 \otimes x$
		\item $\otimes$ distributes over $\vee$.
	\end{enumerate}
\end{definition}
\noindent A quantale is called commutative when its $\otimes$ operator is commutative~\cite{hunt2021quantale}.

Next, we prove some lemmas that are later used in the proof of Theorem~\ref{thrm:determinacy_quantale_proof}.
\begin{globallemma}\label{lemma:app:tc_is_closure}
	Both $\mix$ and $\tc$ are closure operators.
\end{globallemma}
\begin{proof}
	A closure operator is a function $f :\powerset{A} \rightarrow \powerset{A}$ from the power set of domain $A$ to itself that satisfies the following properties for all sets $X,Y \subseteq A$:
	\begin{itemize}
		\item $ f $ is extensive: $X\subseteq \operatorname f(X)$
		\item $ f $ is increasing: $X\subseteq Y\Rightarrow  f(X)\subseteq f(Y)$
		\item $ f $ is idempotent: $f(f(X))= f(X)$
	\end{itemize}
	It is straightforward to show that both $\mix$ and $\tc$ satisfy these conditions.
\end{proof}

\begin{definition}
	For a closure operator $\clEmpty$ defined on the domain $A$, and a function $F : A \rightarrow A$, say that $F$ weakly commutes with $\clEmpty$ if $F(cl(X)) \subseteq cl(F(X))$ for all $X \subseteq A$.
\end{definition}

\begin{globallemma}\label{lemma:app:function_weak_commute}
	Let $\clEmpty : A \rightarrow A$ be a closure operator and let $X, Y \subseteq A$. Suppose that $F : A \rightarrow A$ weakly commutes with $\clEmpty$ and that $G : A \times A \rightarrow A$ weakly commutes with $\clEmpty$ in each argument. Then:
	\begin{enumerate}
		\item $\clEmpty(F (\clEmpty(X))) = \clEmpty(F (X))$
		\item $\clEmpty(G(\clEmpty(X) \times \clEmpty(Y ))) = \clEmpty(G(X \times Y ))$
	\end{enumerate}
\end{globallemma}
\begin{proof}
	Routine, following the properties of closure operator.
\end{proof}

\begin{globallemma}\label{lemma:app:cup_week_commute_tc}
	Let $\P, \Q \subseteq DL(\queries)$, the union operator $\cup$ weakly commutes with $\tc$:
	\begin{align*}
		\tc(\tc(\P) \cup \tc(\Q)) = \tc(\P \cup \Q)
	\end{align*}
\end{globallemma}
\begin{proof}
	It suffices to show $\R \in \tc(\tc(\P) \cup \tc(\Q))$ iff $\R \in \tc(\P \cup \Q)$, which follows easily from the definitions of $\cup$ and $\tc$.
\end{proof}

\begin{globallemma}\label{lemma:app:DL_join_week_commute_tc}
	The join operator of DL weakly commutes with $\tc$ in each argument
\end{globallemma}
\begin{proof}
	Let $P, Q \in DL(\queries)$, and let $\S \subseteq DL(\queries)$. If $\eq{Q}$ is tiled by $\EQ{\S}$ then $\eq{P} \sqcup \eq{Q}$ is tiled by $\{ \eq{P} \sqcup R \mid R \in \EQ{\S} \}$. This follows easily from the definition of the equivalence relation induced by a query (\ie $\eq{}$), $\mix$, Lemma~\ref{lemma:app:join_loi_determinacy_lattice} and the fact that $[\eq{P} \sqcup \eq{Q}] = \{ A \cap B \mid A \in [\eq{P}], B \in [\eq{Q}] \} \setminus \varnothing$.
\end{proof}

\begin{globallemma}\label{lemma:app:tensor_week_commute_tc}
	Given two sets of sets of queries $\Q, \P \subseteq DL(\queries)$ it holds that:
	\begin{align*}
		\tc(\Q) \otimes \tc(\P) = \Q \otimes \P
	\end{align*}
\end{globallemma}
\begin{proof}
	By Lemma~\ref{lemma:app:DL_join_week_commute_tc} we know that the join operator of DL weakly commutes with $\tc$ in each argument. We apply this lemma to the definition of $\otimes$ operator:
	\begin{align*}
		&\tc(\Q) \otimes \tc(\P) = \\
		&\tc( \bigcup_{Q \in \tc(\Q), P \in \tc(\P)} (Q \sqcup P) ) = \\
		&\tc( \bigcup_{Q \in \Q, P \in \P} (Q \sqcup P) ) = \\
		&\Q \otimes \P
	\end{align*}
\end{proof}

\begin{globallemma}\label{lemma:app:DL_join_commutativity}
	Given two sets of sets of queries $\Q, \P \subseteq DL(\queries)$ it holds that:
	\begin{align*}
		\Q \vee \P = \P \vee \Q
	\end{align*}
\end{globallemma}
\begin{proof}
	Follows directly from the definition of $\vee$ in the DL and the commutativity of union operator $\cup$.
	\begin{align*}
		&\Q \vee \P = \\
		&\cl{(\Q \cup \P)} = \\
		&\cl{(\P \cup \Q)} = \\
		&\P \vee \Q
	\end{align*}
\end{proof}

Now, we show that DQ in Def.~\ref{def:determinacy_quantale} is a quantale.
\begin{globaltheorem}\label{thrm:determinacy_quantale_proof}
	The Determinacy Quantale is a commutative quantale.
\end{globaltheorem}
\begin{proof}
	We have to show that our definition of Determinacy Quantale respects the quantale axioms of Def.~\ref{def:quantale_axioms}.
	\begin{enumerate}
		\item Showing $\tuple{\mathcal{I}, \sqsubseteq, \vee}$ is a complete join-semilattice is straightforward following Lemma~\ref{lemma:app:tc_is_closure} and the fact that $\tc$ is a closure operator.
		
		\item We should show that $\otimes$ is associative and $1$ is a unit:
		\begin{enumerate}
			\item[a.] For the associativity of $\otimes$ we need to show that $\P \otimes (\Q \otimes \R) = (\P \otimes \Q) \otimes \R$. 
			Here we rely on Lemmas~\ref{lemma:app:DL_join_week_commute_tc} and~\ref{lemma:app:tensor_week_commute_tc} to eliminate the nested uses of $\tc$ and the basic properties of $\cup$ operator to show that both sides of $\P \otimes (\Q \otimes \R) = (\P \otimes \Q) \otimes \R$ can be reduced to identical expressions. \\
			Left side:
			\begin{align*}
				&\P \otimes (\Q \otimes \R) = \\
				&\P \otimes \tc( \bigcup_{Q \in \Q, R \in \R} (Q \sqcup R) ) = \\
				&\P \otimes ( \bigcup_{Q \in \Q, R \in \R} (Q \sqcup R) ) = \\
				&\tc( \bigcup_{P \in \P, T \in ( \bigcup_{Q \in \Q, R \in \R} (Q \sqcup R) )} (P \sqcup T)  )  = \\
				&\tc( \bigcup_{P \in \P, Q \in \Q, R \in \R)} (P \sqcup Q \sqcup R)  ) \tag{1}
			\end{align*}
			Right Side:
			\begin{align*}
				&(\P \otimes \Q) \otimes \R = \\
				&\tc( \bigcup_{P \in \P, Q \in \Q} (P \sqcup Q))  \otimes \R  = \\
				&(\bigcup_{P \in \P, Q \in \Q} (P \sqcup Q))  \otimes \R  = \\
				&\tc( \bigcup_{T \in ( \bigcup_{P \in \P, Q \in \Q} (P \sqcup Q) ), R \in \R} (T \sqcup R)  )  = \\
				&\tc( \bigcup_{P \in \P, Q \in \Q, R \in \R)} (P \sqcup Q \sqcup R)  ) \tag{2}
			\end{align*}
			By (1) and (2) we can conclude that $\P \otimes (\Q \otimes \R) = (\P \otimes \Q) \otimes \R$.
			
			\item[b.] To show that $1 = \varnothing$ is a unit for $\otimes$ we need to show that $\forall x \in \mathcal{I}, x \otimes 1 = x = 1 \otimes x$. Using $\varnothing$ as the unit, and applying the definition of $\otimes$ will give us:
			\begin{align*}
				\Q \otimes \varnothing = \tc( \bigcup_{Q \in \Q} (Q) ) = \tc(\Q) = \Q
			\end{align*}
			which following the associativity of $\otimes$ gives us $\forall x \in \mathcal{I}, x \otimes \varnothing = x = \varnothing \otimes x$.
		\end{enumerate}
		
		\item To establish distributivity we need to show that $\P \otimes (\Q \vee \R) = (\P \otimes \Q) \vee (\P \otimes \R)$.
		We again rely on Lemmas~\ref{lemma:app:DL_join_week_commute_tc} and~\ref{lemma:app:tensor_week_commute_tc} and basic properties of $\cup$ to show:
		\begin{align*}
			&\P \otimes (\Q \vee \R) = \\
			&\P \otimes \tc(\Q \cup \R) = \\
			&\P \otimes (\Q \cup \R) = \\
			&\tc( \bigcup_{P \in \P, T \in (\Q \cup \R)} (P \sqcup T) ) = \\
			&\tc( \bigcup_{P \in \P, Q \in \Q} (P \sqcup Q) \cup \bigcup_{P \in \P, R \in \R} (P \sqcup R)) = \\
			&\tc( (\P \otimes \Q) \cup (\P \otimes \R)) = \\
			&(\P \otimes \Q) \vee (\P \otimes \R)
		\end{align*}
		
		\item Commutativity of $\otimes$ is inherited directly from Lemma~\ref{lemma:app:DL_join_commutativity} and the commutativity of $\sqcup$ in DL.
		\begin{align*}
			&\P \otimes \Q = \\
			&\tc( \bigcup_{P \in \P, Q \in \Q} (P \sqcup Q) ) = \\
			&\tc( \bigcup_{P \in \P, Q \in \Q} (Q \sqcup P) ) = \\
			&\Q \otimes \P
		\end{align*}
	\end{enumerate}
\end{proof}

\section{Relation Between DQ and QoI}\label{sec:app:proof:DQ_imply_QoI}
We first provide some auxiliary lemmas, and then proceed to prove Lemma~\ref{lemma:DQ_imply_QoI}.

\begin{globallemma}\label{lemma:app:ordering_DQ_QoI}
	Given sets of sets of queries $\Q, \P \subseteq DL(\queries)$, $\tc(\Q) \subseteq \tc(\P)$ on the DQ implies $\EQ{\tc(\Q)} \subseteq \EQ{\tc(\P)}$ on the QoI defined on $\{\EQ{\Q} \mid \Q \subseteq DL(\queries)\}$:
\end{globallemma}
\begin{proof}
	Trivial from the Def.~\ref{def:set_query_to_set_eq}.
\end{proof}

\begin{globallemma}\label{lemma:app:join_DQ_QoI}
	$\bigvee_i \P_i$ on the DQ implies $\bigvee_i \EQ{\P_i}$ on the QoI defined on $\{\EQ{\Q} \mid \Q \subseteq DL(\queries)\}$.
\end{globallemma}
\begin{proof}
	Trivial from the Def.~\ref{def:set_query_to_set_eq}.
\end{proof}

\begin{globallemma}\label{lemma:app:tensor_DQ_QoI}
	Given sets of sets of queries $\Q, \P \subseteq DL(\queries)$, $\tc(\Q) \otimes \tc(\P)$ on the DQ implies $\EQ{\tc(\Q)} \otimes \EQ{\tc(\P)}$ on the QoI defined on $\{\EQ{\Q} \mid \Q \subseteq DL(\queries)\}$.
\end{globallemma}
\begin{proof}
	Follows trivially from Def.~\ref{def:set_query_to_set_eq} and Lemma~\ref{lemma:app:join_loi_determinacy_lattice}.
\end{proof}

\DQimplyQoI*
\begin{proof}
	To prove this homomorphism, we need to show that the Determinacy Quantale's ordering, join and tensor, as well as the top and bottom elements imply their QoI counterparts. Lemmas~\ref{lemma:app:ordering_DQ_QoI}, \ref{lemma:app:join_DQ_QoI}, and \ref{lemma:app:tensor_DQ_QoI} provide the proofs of ordering, join and tensor, respectively. 
	The proof of the top element: 
	\begin{itemize}
		\item $DL(\queries) \rightarrow LoI(\EQ{\queries})$
	\end{itemize}
	follows from Def.~\ref{def:set_query_to_set_eq} and Lemma~\ref{lemma:DL_imply_LoI}, and the proof of the bottom element:
	\begin{itemize}
		\item $\varnothing \rightarrow \varnothing$
	\end{itemize}
	is trivial.
\end{proof}

\section{Correctness of Dependency Analysis} \label{app:dep_analysis}
\def\mix{\mathrm{mix}}
\def\QL{\mathrm{QL}}
\def\interp#1{[{#1}_\sim]}
To show that the diagram in Fig.~\ref{fig:proof_steps} commutes, we aim to show commutativity for each cell in it.
In this section, we establish this for the bottommost cell of it. To that end, we need to establish that the QoI point $\llbracket \mathrm{QL}_u \rrbracket$ that corresponds 
to the query list $\mathrm{QL}_u=\{Q_1,\ldots, Q_n\}$ extracted from a program $\prg$ by the dependency analysis 
is an upper bound on the knowledge relation $\llbracket \prg_u \rrbracket$ induced
by $\prg$. 

The basic outline of the argument rests on identifying a
particular \emph{single} equivalence relation $k(\mathrm{QL}, \prg)\in \mix(\interp{Q_1},\ldots,\interp{Q_n})$,
which satisfies $\llbracket \prg \rrbracket \sqsubseteq k(\mathrm{QL}, \prg)$. Intuitively, this relation captures how much information the program could leak at most if it output the full result of every query that its output depends on. As long as the analysis is sound, this is an instantiation of the disjunction represented by QL, with each disjunct selected precisely for those starting configurations where the program's output turns out to depend on the queries enumerated in that disjunct.

For a fixed program $\prg$ and user $u$, we assume the existence of a function $Q=Q_{\prg,u}$ from databases $db\in \Omega_D$ to sets of queries, which returns the set of those queries performed when executing $\prg$ on database $db$
whose result taints some output to the user $u$. 
We formally define the function $Q$ by relying on a taint analysis.
\def\reset{\texttt{set } pc \texttt{ to }\Delta(pc)}
\def\resetd{\texttt{set } pc \texttt{ to }\delta}
\newcommand{\doubledownharpoon}{\mathrel{\downharpoonleft\mkern-5.7mu\downharpoonleft}}
\newcommand{\projB}[2]{#1\negthickspace\doubledownharpoon_{#2}}

\tightpar{Taint analysis}%
The semantics of the taint analysis enriches the normal operational semantics of the language in the sense that it has transitions whenever the operational one does, and acts the same on those components of a configuration that exist in the operational one; so runs in it can be put in one-to-one correspondence to operational ones.

\begin{figure*}[h]
	{
		\setstretch{1.2}
		\footnotesize
		\centering
		$
		\inferrule*[before=\textsc{TA-Skip}]
		{
			\\
		}
		{
			\tuple{\Delta, \texttt{skip}, m, db} \xrightarrow{\epsilon} \tuple{\Delta, \epsilon, m, db}
		}
		$
		
		\nextrule
		
		$
		\inferrule*[before=\textsc{TA-Assign}]
		{
			\tuple{e, m, db}  \downarrow \vl \\
			m' = m[x \mapsto \vl] \\
			\Delta' = \Delta[x \mapsto \Delta(pc) \cup \textstyle \bigcup_{x \in \fv(e)} \Delta(x)]
		}
		{
			\tuple{\Delta, x := e, m, db} \xrightarrow{\epsilon} \tuple{\Delta', \epsilon, m', db}
		}
		$
		
		\nextrule
		
		$
		\inferrule*[before=\textsc{TA-QueryEval}]
		{
			\vl = \eval{q}{db} \\
			m' = m[x \mapsto \vl] \\
			\Delta' = \Delta[x \mapsto \Delta(pc) \cup q]
		}
		{
			\tuple{\Delta, x \leftarrow q, m, db} \xrightarrow{\epsilon} \tuple{\Delta', \epsilon, m', db}
		}
		$
		
		\nextrule
		
		$
		\inferrule*[before=\textsc{TA-IfTrue}]
		{
			\tuple{e, m, db} \downarrow n \\
			n \not= 0 \\\\
			c'_1 = c_1 ; \reset \\\\
			\Delta' = \Delta[pc \mapsto \Delta(pc) \cup \textstyle \bigcup_{x \in \fv(e)} \Delta(x)]
		}
		{
			\tuple{\Delta, \texttt{if} \ e \ \texttt{then} \ c_1 \ \texttt{else} \ c_2, m, db} \xrightarrow{\epsilon} \tuple{\Delta', c'_1, m, db}
		}
		$
		\nextrule
		$
		\inferrule*[before=\textsc{TA-IfFalse}]
		{
			\tuple{e, m, db} \downarrow n \\
			n = 0 \\\\
			c'_2 = c_2 ; \reset \\\\
			\Delta' = \Delta[pc \mapsto \Delta(pc) \cup \textstyle \bigcup_{x \in \fv(e)} \Delta(x)]
		}
		{
			\tuple{\Delta, \texttt{if} \ e \ \texttt{then} \ c_1 \ \texttt{else} \ c_2, m, db} \xrightarrow{\epsilon} \tuple{\Delta', c'_2, m, db}
		}
		$
		
		\nextrule
		
		$
		\inferrule*[before=\textsc{TA-WhileTrue}]
		{
			\tuple{e, m, db}  \downarrow n \\
			n \not= 0 \\\\
			c' = c;\texttt{while} \ e \ \texttt{do} \ c ; \reset \\\\
			\Delta' = \Delta[pc \mapsto \Delta(pc) \cup \textstyle\bigcup_{x \in \fv(e)} \Delta(x)]
		}
		{
			\tuple{\Delta, \texttt{while} \ e \ \texttt{do} \ c, m, db} \xrightarrow{\epsilon} \tuple{\Delta', c', m, db}
		}
		$		
		\nextrule
		$
		\inferrule*[before=\textsc{TA-WhileFalse}]
		{
			\tuple{e, m, db} \downarrow n \\
			n = 0 \\\\
			c' = \reset \\\\
			\Delta' = \Delta[pc \mapsto \Delta(pc) \cup \textstyle \bigcup_{x \in \fv(e)} \Delta(x)]
		}
		{
			\tuple{\Delta, \texttt{while} \ e \ \texttt{do} \ c, m, db} \xrightarrow{\epsilon} \tuple{\Delta', \epsilon, m, db}
		}
		$	
		
		\nextrule
		
		$
		\inferrule*[before=\textsc{TA-Seq}]
		{
			\tuple{\Delta, c_1, m, db} \xrightarrow{\alpha} \tuple{\Delta', c_{1}', m', db'} \\
		}
		{
			\tuple{\Delta, c_1;c_2, m, db} \xrightarrow{\alpha} \tuple{\Delta', c_{1}';c_2, m', db'}
		}
		$
		\nextrule
		$
		\inferrule*[before=\textsc{TA-SeqEmpty}]
		{
			\\
		}
		{
			\tuple{\Delta, \epsilon;c, m, db} \xrightarrow{\epsilon} \tuple{\Delta, c, m, db}
		}
		$
		
		\nextrule
		
		$
		\inferrule*[before=\textsc{TA-Output}]
		{
			\tuple{e, m, db}  \downarrow \vl \\
			\beta = \Delta(pc) \cup \textstyle \bigcup_{x \in \fv(e)} \Delta(x)
		}
		{
			\tuple{\Delta, \texttt{out}(e,u), m, db} \xrightarrow{\tuple{\vl, u, \beta}} \tuple{\Delta, \epsilon, m, db}
		}
		$
		\nextrule
		$
		\inferrule*[before=\textsc{TA-SetPC}]
		{
			\Delta' = \Delta[pc \mapsto \delta]
		}
		{
			\tuple{\Delta, \resetd, m, db} \xrightarrow{\epsilon} \tuple{\Delta', \epsilon, m, db}
		}
		$
		
	}
	\caption{Taint analysis rules}
	\label{fig:taint_analysis_rules}
\end{figure*}

The rules of the taint analysis presented in Fig.~\ref{fig:taint_analysis_rules} are fairly straightforward. We use mapping $\Delta$ to map each variable to a set of dependencies of variables and queries.

The rules for \texttt{if} rely on auxiliary command $\resetd$ to restore the dependency set of $pc$ to its previous state ($\Delta(pc)$) upon exiting the \texttt{if} branch. We sequentially composite this command with the body of \texttt{if} to ensure its execution after leaving the \texttt{if} branch's body. The rules for \texttt{while} use $\resetd$ in a similar manner.

The rule \textsc{TA-Output} uses $\fv(e)$ to extract all the variables of expression $e$, and relies on the union of the $\Delta$s of those variables to calculate $\beta$, which is the set of dependencies the execution up to this output, depended on. 

We extend the definition of trace $\tau$ to a sequence of observations of the form $\tuple{\vl, u, \beta}$, and use the notation $\projB{\tau}{u}$ to denote the sequence of all $\beta$s in $\tau$ that $u$ can observe. We use this notation to define function $Q$ as follows:
\begin{definition}\label{def:function_Q}
	Given a database state $db$ and user $u$, such that $\tuple{c, \initM, db} \mulStepEval{\tau}{u}$, $Q(db)$ is defined as $\{ \beta \mid \beta \in \projB{\tau}{u}\}$
\end{definition}

A proof of Lemma \ref{lemma:type_analysis_soundness_proof} can then proceed by a straightforward induction on the semantics.

In Def.~\ref{def:function_Q} we formally define the function $Q$.
This function satisfies a closure property that informally states that if on
two given databases the output depended on different sets of queries,
then the choice of the set of dependencies itself must have been due
to the outcome of a query which is among the dependencies in both databases
and evaluates to a different result. 
\begin{globallemma}\label{lem:taintclosed}
 For all $db,db'\in \Omega_D$, if $Q(db)\neq Q(db')$, then there exists
a particular query $q\in Q(db)\cap Q(db')$ such that $\eval{q}{db} \neq \eval{q}{db'}$.
\end{globallemma}

We say that database states $db$ and $db'$ are equivalent with respect to a dependency set $S$ (written as $db \approx_{S} db'$) iff $\eval{y}{db} = \eval{y}{db'}$ for all $y \in S$ where $y \in \queries$.
\begin{globallemma}\label{lem:taintsoutputs}
 For all states $db_1$ and $db_2$ and users $u$, if $\tuple{c, \initM, db_1} 
\mulStepEval{t_1}{u}$, %
 $\tuple{c, \initM, db_2} \mulStepEval{t_2}{u}$,
$Q\triangleq Q(db_1)=Q(db_2)$ and $db_1 \approx_{Q} db_2$, then $\proj{t_1}{u} = \proj{t_2}{u}$.
\end{globallemma}

We then define $k(\QL_u, \prg)$ as the equivalence relation
$$ \{ (db,db') \in \Omega_D^2 \mid Q(db)=Q(db') \wedge (db,db') \in Q(db)_\sim \}, $$
that is, we partition each respective subset of databases $db$ that shares one set of queries $Q(db)$ into equivalence classes according to the knowledge relation induced by $Q(db)$.
\begin{globallemma} $\llbracket \prg \rrbracket_u \sqsubseteq k(\QL_u, \prg) \sqsubseteq \llbracket \QL_u \rrbracket$. 
\end{globallemma}
\begin{proof}
$k(\QL_u, \prg) \sqsubseteq \llbracket \QL_u \rrbracket$: Will in fact show
that $k(\QL_u, \prg) \in \mix(\interp{Q_1},\ldots,\interp{Q_n})$, where $\QL_u=\{Q_1,\ldots,Q_n\}$.
For that, it suffices to show that every equivalence class $x\in [k(\QL_u, \prg)]$ is also 
an equivalence class of one of the $Q_i$. Let $db\in x$ be arbitrary. Then claim that $x\in [Q(db)]$,
which suffices since by Lemma \ref{lemma:type_analysis_soundness_proof}, $Q(db)$ is one of the $Q_i$. To establish this, just need 
to show that $Q(db)=Q(db')$ for all $db'$ such that $(db,db')\in Q(db)_\sim$,
so that $(db,db')\in k(\QL_u, \prg)$ as well. But this follows from Lemma \ref{lem:taintclosed}: if
some $db'$ has $(db,db')\in Q(db)_\sim$, then $\eval{q}{db} = \eval{q}{db'}$ for all $q\in Q(db)$,
but then we must not have $Q(db)\neq Q(db')$.

$\llbracket \prg \rrbracket_u \sqsubseteq k(\QL_u, \prg)$: Straightforward application of Lemma \ref{lem:taintsoutputs}.
\end{proof}

\section{Query Analysis}\label{sec:app:proof:query_analysis_proofs}

\subsection{Symbolic Tuple Ordering}\label{sec:app:proof:symbolic_tuple_ordering}
To show that the symbolic tuples ordering of Def.~\ref{def:symbolic_tuples_ordering} induces a determinacy order and prove Lemma~\ref{lemma:symbolic_tuple_ordering_determinacy_proof} we first need to define the evaluation of a symbolic tuple in a database state.

\tightpar{Symbolic tuple evaluation}
The evaluation of a symbolic tuple $\tuple{T,\phi, \pi}$ in the database state $db$ written as $\eval{\tuple{T, \phi, \pi}}{db}$ is a $\pi$-projection on the set of $db$'s tuples defined on the join of tables in $T$ that satisfy the constraint $\phi$. Formally:

\begin{definition}\label{def:symbolic_tuples_evaluation}
	Given database state $db$ and symbolic tuple $\tuple{T,\phi, \pi}$, $\eval{\tuple{T,\phi, \pi}}{db}$ is defined as:
	\begin{align*}
		\{ \proj{\tp}{\pi} \mid \tp \in \prod_{\tb \in T} \evalTable{\tb}{db}, \tp \models \phi \}
	\end{align*}
	where $\proj{\tp}{\pi}$ is a tuple with its columns limited to those in $\pi$, and $\tp \models \phi$ means that tuple $\tp$ satisfies formula $\phi$.
\end{definition}

We proceed to prove Lemma~\ref{lemma:symbolic_tuple_ordering_determinacy_proof}.
\begin{restatable}{globallemma}{symbolicTupleOrderingDeterminacyProof}\label{lemma:symbolic_tuple_ordering_determinacy_proof}
	Given two sets of queries $Q_1$ and $Q_2$, if $\mathrm{sts}(Q_1) \stsubset \mathrm{sts}(Q_2)$ then $Q_1 \preceq Q_2$.
\end{restatable}

\begin{proof}
	Assume $\ell_{Q_1} = \mathrm{sts}(Q_1)$ and $\ell_{Q_2} = \mathrm{sts}(Q_2)$. By Def.~\ref{def:symbolic_tuples_ordering} we want to show that if for all symbolic tuples $\tuple{T,\phi, \pi} \in \ell_{Q_1}$, there is a set of well-formed symbolic tuples $S = \tuple{T_1,\phi_1, \pi_1},...,\tuple{T_n,\phi_n, \pi_n}$ such that $S \subseteq \ell_{Q_2}$, $T_1,..., T_n$ are disjoint, $T \subseteq (T_1 \cup ... \cup T_n)$, $\phi \models (\phi_1 \wedge ... \wedge \phi_n)$, and $\dep(\phi) \cup \pi \subseteq (\pi_1 \cup ... \cup \pi_n)$, then $Q_1 \preceq Q_2$.
	
	We assume an intermediate symbolic tuple $\st_{\itr}$ and define it as $\tuple{T_1 \cup ... \cup T_n, \phi_1 \wedge ... \wedge \phi_n, \pi_1 \cup ... \cup \pi_n}$. $\st_{\itr}$ models the symbolic tuples created from the join of $\tuple{T_1,\phi_1, \pi_1},...,\tuple{T_n,\phi_n, \pi_n}$. 
	Additionally, $T_1,..., T_n$ are disjoint, which by the definition of symbolic tuples means that $\pi_1, ..., \pi_n$ and the dependencies of $\phi_1, ..., \phi_n$ are also disjoint, effectively making $\st_{\itr}$ the symbolic tuple of the Cartesian product of tuples $\tuple{T_1,\phi_1, \pi_1},...,\tuple{T_n,\phi_n, \pi_n}$.
	
	We want to show that the symbolic tuples in $S$ can determine $\st_{\itr}$:
	\begin{align*}
		&\forall db_1, db_2 \in \Omega_D. \\
		&\eval{\st}{db_1} = \eval{\st}{db_2} \ \forall \st \in S \rightarrow \eval{\st_{\itr}}{db_1} = \eval{\st_{\itr}}{db_2} \tag{1}
	\end{align*}
	
	For a specific database state $db$, $\eval{\st_{\itr}}{db}$ would give us all the tuples defined on $T_1,..., T_n$ satisfying $\phi_1 \wedge ... \wedge \phi_n$ and projected on the columns in $\pi_1 \cup ... \cup \pi_n$.
		
	Assume there is a pair of databases $db_1, db_2 \in \Omega_D$ such that $\eval{\st}{db_1} = \eval{\st}{db_2} \ \forall \st \in S$ holds but $\eval{\st_{\itr}}{db_1} \not= \eval{\st_{\itr}}{db_2}$. 
	By the assumption $\eval{\st}{db_1} = \eval{\st}{db_2} \ \forall \st \in S$ we know that for all $\st \in S$, if tuple $\tp$ is in $\eval{\st}{db_1}$ it is also in $\eval{\st}{db_2}$, and vice versa.
	
	For $\eval{\st_{\itr}}{db_1} \not= \eval{\st_{\itr}}{db_2}$ to hold, we have to consider two cases:
	\begin{enumerate}
		\item There is a tuple $\tp \in \eval{\st_{\itr}}{db_1}$ such that $\tp$ cannot be constructed from the tuples in set $\{ \tp' \in \eval{\st}{db_1} \mid \st \in S\}$
		\begin{itemize}
			\item[-] All of the symbolic tuples $\st \in S$ are well-formed and $T_1,..., T_n$ are disjoint, which makes $\st_{\itr}$ the symbolic tuple of the Cartesian product of $S$. This means that tuple $\tp \in \eval{\st_{\itr}}{db_1}$ is defined on the product of tables $T_1,..., T_n$, satisfies $\phi_1 \wedge ... \wedge \phi_n$, and projected on $\pi_1 \cup ... \cup \pi_n$. Which means that each tuple $\tp \in \eval{\st_{\itr}}{db_1}$ is constructed from the merge of tuples $\tp_1, ..., \tp_n$ where $\tp_i \in \eval{\tuple{T_i,\phi_i, \pi_i}}{db_1}$ for $i = 1,...,n$. 
			Thus, this case is not possible.
		\end{itemize}
		
		\item There is a tuple $\eval{\st_{\itr}}{db_2}$ such that $\tp$ cannot be constructed from the tuple set $\{ \tp' \in \eval{\st}{db_2} \mid \st \in S\}$
		\begin{itemize}
			\item[-] Similar to the first case.
		\end{itemize}
	\end{enumerate}
	
	Next, we need to show that $\st_{\itr}$ determines $\tuple{T,\phi, \pi}$: 
	\begin{align*}
		&\forall db_1, db_2 \in \Omega_D \\
		&\eval{\st_{\itr}}{db_1} = \eval{\st_{\itr}}{db_2} \rightarrow \eval{\tuple{T,\phi, \pi}}{db_1} = \eval{\tuple{T,\phi, \pi}}{db_2} \tag{2}
	\end{align*}
	
	By $\eval{\st_{\itr}}{db_1} = \eval{\st_{\itr}}{db_2}$ we know $\forall \tp_1 \in \eval{\st_{\itr}}{db_1}, \exists \tp_2 \in \eval{\st_{\itr}}{db_2}$ and $\tp_1 = \tp_2$, and $\forall \tp_2 \in \eval{\st_{\itr}}{db_2}, \exists \tp_1 \in \eval{\st_{\itr}}{db_1}$ and $\tp_2 = \tp_1$.
	
	Intuitively, for a given database $db$, $\eval{\st_{\itr}}{db}$ has has more columns and tuples than $\eval{\tuple{T,\phi, \pi}}{db}$. Symbolic tuple $\tuple{T,\phi, \pi}$ throws away some columns by limiting the resulting tuples to tables in $T$ which is a subset of $T_1 \cup ... \cup T_n$ and projecting on $\pi$ which is a subset of $\pi_1 \cup ... \cup \pi_n$. It also eliminate some rows by applying $\phi$ to the result set, which is stronger than $\phi_1 \wedge ... \wedge \phi_n$.

	We need to show that applying these limitations maintains query determinacy. We consider these cases separately:
	\begin{itemize}
		\item \textbf{Columns}: Projecting away some columns from the evaluation of $\st_{\itr}$ is going to maintain query determinacy.
		We denote by $\proj{\tp}{\pi}$, projecting tuple $\tp$ to only columns specified in $\pi$, additionally we use the notation $\col(T)$ to indicate the columns of $T$. We use the same notation for tuples and write $\col(\tp)$ to denote the set of columns of tuple $\tp$.
		For a tuple $\tp$ such that $\tp \in \eval{\st_{\itr}}{db_1}$ and $\tp \in \eval{\st_{\itr}}{db_2}$, by projecting away some columns from $\tp$ we end up with a new tuple $\tp' = \proj{\tp}{\pi}$ such that $\col(\tp') \subseteq \col(\tp)$. 
		Since $\tp$ is in both $\eval{\st_{\itr}}{db_1}$ and $\eval{\st_{\itr}}{db_2}$, and by the definition of ordering $\pi \subseteq \pi_1 \cup ... \cup \pi_n$, we can conclude $\tp'$ will also be in both $\eval{\tuple{T_1 \cup ... \cup T_n,\phi_1 \wedge ... \wedge \phi_n, \pi}}{db_1}$ and $\eval{\tuple{T_1 \cup ... \cup T_n,\phi_1 \wedge ... \wedge \phi_n, \pi}}{db_2}$, this follows easily from Def.~\ref{def:symbolic_tuples_evaluation}.
		
		\item \textbf{Rows}: Removing some rows from the last step is going to maintain query determinacy.
		By the definition of ordering we know that $\dep(\phi) \cup \pi \subseteq \pi_1 \cup ... \cup \pi_n$ and that $\tuple{T_1,\phi_1, \pi_1},...,\tuple{T_n,\phi_n, \pi_n}$ are well-formed, which means that $\phi$ only applies to the columns that were retrieved by the intermediate tuple (projected to $\pi$). 
		Since $\phi$ is a stronger condition than $\phi_1 \wedge ... \wedge \phi_n$, for a tuple $\tp$ such that $\tp \in \eval{\tuple{T_1 \cup ... \cup T_n,\phi_1 \wedge ... \wedge \phi_n, \pi}}{db_1}$ and $\tp \in \eval{\tuple{T_1 \cup ... \cup T_n,\phi_1 \wedge ... \wedge \phi_n, \pi}}{db_2}$, if $\tp$ satisfies $\phi$ then $\tp$ would also be in both $\eval{\tuple{T_1 \cup ... \cup T_n,\phi, \pi}}{db_1}$ and $\eval{\tuple{T_1 \cup ... \cup T_n,\phi, \pi}}{db_2}$. Otherwise, if $\tp$ is not in one of then, it is not going to be in the other one either.
		
		\item \textbf{Tables}: %
		Similar to the first case, for a tuple $\tp$ such that $\tp \in \eval{\tuple{T_1 \cup ... \cup T_n,\phi, \pi}}{db_1}$ and $\tp \in \eval{\tuple{T_1 \cup ... \cup T_n,\phi, \pi}}{db_2}$, by projecting away the columns of some of the tables from $\tp$ we end up with a new tuple $\tp' = \proj{\tp}{\col(T)}$. 
		Since $\tp$ is in both $\eval{\tuple{T_1 \cup ... \cup T_n,\phi_1 \wedge ... \wedge \phi_n, \pi}}{db_1}$ and $\eval{\tuple{T_1 \cup ... \cup T_n,\phi_1 \wedge ... \wedge \phi_n, \pi}}{db_2}$, and by Def.~\ref{def:symbolic_tuples_ordering} $T \subseteq T_1 \cup ... \cup T_n$, we can conclude $t'$ will also be in both $\eval{\tuple{T,\phi, \pi}}{db_1}$ and $\eval{\tuple{T,\phi, \pi}}{db_2}$.
	\end{itemize}
	
	(1) and (2) would give us:
	\begin{align*}
		&\forall db_1, db_2 \in \Omega_D \\
		&\eval{\st}{db_1} = \eval{\st}{db_2} \ \forall \st \in S \rightarrow 
		\eval{\tuple{T,\phi, \pi}}{db_1} = \eval{\tuple{T,\phi, \pi}}{db_2}
	\end{align*}
	which allows us to conclude $\tuple{T_1,\phi_1, \pi_1} ... \tuple{T_n,\phi_n, \pi_n}$ determines $\tuple{T,\phi, \pi}$.
	
	Repeating this process for all of the symbolic tuples in $\ell_{Q_1}$ would give us $Q_2 \twoheadrightarrow Q_1$ which means $Q_1 \preceq Q_2$.
\end{proof}

\subsection{Symbolic Tuple and DQ Ordering}\label{sec:app:proof:symbolic_tuple_DQ_ordering}
We present the proof of Lemma~\ref{lemma:unfolding_and_quantale}.

\unfoldingAndQuantaleProof*

\begin{proof}
	Assume $\sigma_{\st}(\{Q_1, ..., Q_n\}) = \{\ell_{Q_1}, ..., \ell_{Q_n}\}$ and $\sigma_{\st}(\{P_1, ..., P_m\}) = \{\ell_{P_1}, ..., \ell_{P_m}\}$. We have $\{\ell_{Q_1}, ..., \ell_{Q_n}\}$ $\sqsubseteq_* \{\ell_{P_1}, ..., \ell_{P_m}\}$. 
	
	By the definition of $\sqsubseteq_*$ and Lemma~\ref{lemma:symbolic_tuple_ordering_determinacy_proof}, we know that for each $Q_i$ in $\{Q_1, ..., Q_n\}$ there is at least one $P_j$ in $\{P_1, ..., P_m\}$ such that $Q_i \preceq P_j$.
	
	We apply $\tc$ to $Q_i$ and $P_j$ which would give us $\tc(Q_i) \subseteq \tc(P_j)$. 
	By applying $\tc$ to every element of $\{P_1, ..., P_m\}$, using the basic properties of $\cup$ we will have $\tc(Q_i) \subseteq \tc(P_1) \cup ... \cup \tc(P_m)$ for all $i \in \{1, ..., n\}$.
	
	Since the tiling closure of each $Q_i$ is individually less than $\tc(P_1) \cup ... \cup \tc(P_m)$, their union would still be less that $\tc(P_1) \cup ... \cup \tc(P_m)$ which gives us:
	\begin{align}
		\tc(Q_1) \cup ... \cup \tc(Q_n) \subseteq \tc(P_1) \cup ... \cup \tc(P_m) \tag{1}
	\end{align}
	We apply the tiling closure to both sides of (1) and rely on Lemma~\ref{lemma:app:cup_week_commute_tc} to remove the nested uses of $\tc$, which would give us:
	\begin{align*}
		\tc(Q_1 \cup ... \cup Q_n) \subseteq \tc(P_1 \cup ... \cup P_m)
	\end{align*}
	which by the definition of $\vee$ in the DQ would mean $(Q_1 \vee ... \vee Q_n) \sqsubseteq (P_1 \vee ... \vee P_m)$
\end{proof}

\newpage
\else
\fi

\end{document}